\documentclass{article}
\usepackage{import}
\usepackage{basic}

\title{Bilateral Trade with Correlated Values}
\date{}

\begin{document}

\author{Shahar Dobzinski \and Ariel Shaulker\thanks{Weizmann Institute of Science.  Emails: {\ttfamily\{shahar.dobzinski, ariel.shaulker\}@weizmann.ac.il}. Work supported by ISF grant 2185/19 and BSF-NSF grant (BSF number: 2021655, NSF number: 2127781).}}

\maketitle

\begin{abstract}
We study the bilateral trade problem where a seller owns a single indivisible item, and a potential buyer seeks to purchase it. Previous mechanisms for this problem only considered the case where the values of the buyer and the seller are drawn from independent distributions. In contrast, this paper studies bilateral trade mechanisms when the values are drawn from a joint distribution.

We prove that the buyer-offering mechanism guarantees an approximation ratio of $\frac e {e-1} \approx 1.582$ to the social welfare even if the values are drawn from a joint distribution. The buyer-offering mechanism is Bayesian incentive compatible, but the seller has a dominant strategy. We prove the buyer-offering mechanism is optimal in the sense that no Bayesian mechanism where one of the players has a dominant strategy can obtain an approximation ratio better than $\frac e {e-1}$. We also show that no mechanism in which both sides have a dominant strategy can provide any constant approximation to the social welfare when the values are drawn from a joint distribution. 

Finally, we prove some impossibility results on the power of general Bayesian incentive compatible mechanisms. In particular, we show that no deterministic Bayesian incentive-compatible mechanism can provide an approximation ratio better than $1+\frac {\ln 2} 2\approx 1.346$.
\end{abstract}

\thispagestyle{empty}

\newpage
\pagenumbering{arabic}

\section{Introduction}

This paper focuses on the bilateral trade problem where a seller owns a single indivisible item, and a potential buyer seeks to purchase it. The seller has a value of $s\geq 0$ associated with retaining the item (and $0$ otherwise), whereas the buyer's value for obtaining it is $b\geq 0$ (if the buyer does not receive the item, then the buyer's value is $0$). The values $b$ and $s$ are private, but the probability distributions from which they were derived are known.

The two most common objectives in bilateral trade scenarios are to maximize the social welfare and to maximize the gains from trade. The former goal is aimed at maximizing the total value generated by the transaction (that is, the social welfare is $b$ in case of trade and $s$ otherwise), while the latter is focused on increasing the difference between the buyer's and the seller's surplus (i.e., the gains from trade is $b-s$ in case of trade, and $0$ otherwise). Our interest in this paper is in incentive-compatible mechanisms that are strongly budget balanced. That is, the buyer's payment is fully transferred to the seller. See Section \ref{sec-prelim} for a precise statement of the problem.

The problem was extensively studied, and here, we only mention some of the papers closest to our research. Myerson and Satterthwaite \cite{Myerson-Satterthwaite} prove that under very mild conditions, no Bayesian mechanism can exactly maximize the gains from trade or, equivalently, the social welfare. Blumrosen and Dobzinski \cite{BD14} show that a fixed-price mechanism guarantees at least half of the optimal social welfare (equivalently, provides a $2$-approximation) for every distribution. This approximation ratio was later improved to $1.99$ \cite{italians}, then to $\frac  e {e-1} $ \cite{BD21}, then to $\frac e {e-1} -0.0001$ \cite{KPV}, and then to an almost optimal ratio of approximately $1.38$ \cite{CW-STOC, CSTOC}.

There has also been much work regarding approximating the gains from trade. McAfee \cite{mcaffee} shows that for some distributions of the buyer and seller's values, there exists a fixed price mechanism that recovers half of the optimal gains-from-trade. However, for every fixed fraction, there are distributions for which any fixed price guarantees less than this fraction of the optimal gains-from-trade \cite{BD14}. To overcome this, Blumrosen and Mizrahi \cite{BM} propose the seller-offering mechanism, in which the seller makes the buyer the profit-maximizing take-it-or-leave-it offer, given his value $s$. In the seller-offering mechanism, only the buyer has a dominant strategy. Still, they show that if the buyer's value is drawn from a distribution with a monotone hazard rate, then the seller-offering mechanism is Bayesian incentive compatible and provides an  $e$ approximation to the optimal gains-from-trade. Brustle et al. \cite{Bru} show that the better of the seller offering mechanism and the buyer offering mechanism (in which the buyer makes the seller a profit-maximizing take-it-or-leave-it-offer) recovers at least half of the gains from trade of every incentive-compatible mechanism. In a breakthrough result, Deng et al. \cite{GDTCA} show that the better of the offering mechanisms provides ${8.23}$ approximation to the optimal gains-from-trade. This constant was later improved to ${3.15}$ by Fei \cite{IGFT}. We also know that the better of the offering mechanisms sometimes recovers strictly less than half of the optimal gains from trade \cite{DKB}.

Despite this extensive work, existing research on bilateral trade, including the works cited above, largely assumes that the values of the seller and the buyer are drawn independently\footnote{There are some exceptions. For example, McAfee and Reny \cite{MR} provide conditions on the distribution in which the payment of the buyer equals his value if relaxing the individual rationality condition is allowed. Equilibria in some related information models are analyzed in \cite{corThesis}. Finally, Robust mechanisms for bilateral trade are studied in \cite{malik22}.}. This paper investigates a more realistic and technically challenging scenario in which the values are derived from a joint probability distribution. Remarkably, we demonstrate that despite the correlation, the buyer-offering mechanism approximates the social welfare within a constant factor. Furthermore, this factor is quite close to the best approximation ratio possible for independent distributions.

\vspace{0.1in} \noindent \textbf{Theorem: }For every joint distribution of the seller and buyer's values, there is a mechanism that provides an $\frac e {e-1}$-approximation to the optimal social welfare\footnote{The buyer-offering mechanism achieves this approximation ratio for every distribution for which it is defined. There are distributions for which the buyer-offering mechanism is undefined (i.e., there is a series of offers that approach the maximum profit but never attain it). For these distributions, we show that for every $\varepsilon>0$, there is a slight variation of the buyer-offering mechanism that guarantees an $\frac e {e-1}+\varepsilon$ approximation.}.

\vspace{0.1in} \noindent Interestingly, whereas the power of the buyer-offering mechanism is equal to the power of the seller-offering mechanism for gains-from-trade approximation, the buyer-offering mechanism is much more powerful than the seller-offering mechanism in our setting: it is not hard to see that the seller-offering mechanism does not guarantee any constant fraction of the optimal social welfare\footnote{Consider a seller with a fixed value $s=0$ and a buyer that its value is distributed by an equal revenue distribution in $[1,k]$. The optimal welfare is $\ln k$ whereas the welfare of the seller offering mechanism is only $1$.}.

The buyer-offering mechanism is natural and simple. It is also appealing from a theoretical point of view: it is not only Bayesian incentive compatible, but in fact, one player (the offered player) has a dominant strategy as it can only accept or reject a take-it-or-leave-it offer. We call Bayesian incentive compatible mechanisms in which one side has a dominant strategy \emph{one-sided dominant strategy} mechanisms. We prove that the buyer-offering mechanism is optimal among all one-sided mechanisms:

\vspace{0.1in} \noindent \textbf{Theorem: }There exists a joint distribution of the seller and buyer's values such that the approximation ratio of every one-sided dominant strategy mechanism is no better than $\frac e {e-1}$.

\vspace{0.1in} \noindent This theorem demonstrates, in particular, that even taking, e.g., the better of the seller-offering and the buyer-offering mechanism, or the better of the buyer-offering mechanism and some fixed price mechanism, does not improve the approximation guarantee of the buyer-offering mechanism.

Moreover, we show that mechanisms in which both sides have a dominant strategy cannot guarantee any fixed fraction of the optimal social welfare\footnote{A simple observation is that mechanisms in which both sides have a dominant strategy are fixed price mechanisms \cite{BD14}.}. 

\vspace{0.1in} \noindent \textbf{Theorem: }For every constant $c>1$, there exists a joint distribution of the seller and buyer's values such that the approximation ratio of every dominant strategy incentive compatible mechanism is at least $c$.

\vspace{0.1in} \noindent Unlike the impossibility of one-sided mechanisms (which requires careful analysis and subtle construction), the proof that dominant-strategy mechanisms have no power is technically simpler. We then continue analyzing the power of Bayesian incentive-compatible mechanisms, now proving limits on their power:

\vspace{0.1in} \noindent \textbf{Theorem: } 
There exists a joint distribution of the seller and buyer's values such that no deterministic Bayesian incentive-compatible mechanism provides a better than $1+\frac {\ln 2} 2 \approx 1.346$ approximation to the optimal social welfare.

\vspace{0.1in} \noindent Previously, Blumrosen and Mizrahi \cite{BM} proved a bound of $1.07$ on the approximation ratio of Bayesian mechanisms for independent distributions. Their proof relies on characterizing the second-best mechanism of \cite{Myerson-Satterthwaite}. In our setting, proving impossibilities with this approach is less promising: not only is the Myerson-Satterthwaite characterization often hard to compute for independent distributions, but it also does not apply to joint distributions. Thus, we present a new technique for proving impossibilities for Bayesian incentive-compatible mechanisms. Our approach is based on introducing and analyzing a certain family of ``L-shaped'' distributions, for which the second-best mechanism has a nicer structure. A disadvantage of our approach is that our results only apply to deterministic Bayesian mechanisms that are ex-post individually rational, whereas the results of Blumrosen and Mizrahi apply to randomized mechanisms that are interim individually rational ones. We note that all major mechanisms considered in the literature (e.g., fixed price mechanisms, buyer and seller offering mechanisms) are deterministic and ex-post individually rational\footnote{Our impossibilities also apply to mechanisms which are a probability distribution over deterministic ex-post individually rational mechanisms, like the random offerer mechanism \cite{GDTCA}. To see this, consider such a mechanism $A$. A simple averaging argument shows that in the support of $A$ there must be a (deterministic and ex-post individually rational) mechanism $A'$ that its approximation ratio is at least the approximation ratio of $A$. Our impossibilities directly apply to $A'$, and the approximation ratio of $A$ is no better than the approximation ratio of $A'$.}. Nevertheless, our technique is robust enough that, as a by-product, we improve the state-of-the-art impossibilities for both the social welfare and gains-from-trade even with \emph{independent} distributions:

\vspace{0.1in} \noindent \textbf{Theorem: } 
\begin{itemize}
\item There exist independent distributions of the seller and buyer's values such that no deterministic Bayesian incentive-compatible mechanism provides an approximation ratio better than $2$-approximation to the optimal gains from trade.
\item There exist independent distributions of the seller and buyer's values such that no deterministic Bayesian incentive-compatible mechanism provides an approximation ratio better than $1.113$-fraction to the optimal social welfare.
\end{itemize}


\subsubsection*{Open Questions and Future Directions}

In this paper, we analyzed the power of incentive-compatible mechanisms for bilateral trade. We proved that the buyer-offering mechanism provides an approximation ratio of $\frac e {e-1}\approx 1.582$ even if the values are drawn from joint distributions. We proved that this ratio is optimal for one-sided mechanisms and that dominant strategy mechanisms cannot guarantee any fixed fraction of the welfare at all. However, there is a gap between this ratio and our impossibility result of $1+\frac {\ln 2} 2\approx 1.346$. We leave closing this gap as an open question. It will also be interesting to understand the power of Bayesian incentive compatible mechanisms for independent distributions and determine whether they are more powerful than deterministic mechanisms \cite{CSTOC, CW-STOC}.

Another question that remains open is to determine whether Bayesian incentive-compatible mechanisms can give a constant approximation to the gains-from-trade even when the values are drawn from a joint distribution. 

Lastly, all our impossibility results for Bayesian incentive-compatible mechanisms hold only for deterministic Bayesian incentive-compatible mechanisms. In Section \ref{gap-section}, we show that there exist distributions for which randomized Bayesian incentive-compatible mechanisms outperform deterministic Bayesian incentive-compatible mechanisms. An important future direction is understanding the power of randomized Bayesian incentive-compatible mechanisms in all models discussed in the paper.

\paragraph{Structure of the Paper.} In Section \ref{sec-prelim} we give the necessary preliminaries. In Section \ref{sec-algorithm}, we prove that the buyer-offering mechanism provides an $\frac e {e-1}$-approximation, even for correlated distributions. Section \ref{one-sided-lb} shows that $\frac e {e-1}$ is the best ratio achievable by one-sided dominant strategy mechanism.  In Section \ref{sec-impossibilities-bayesian}, we prove several impossibilities for Bayesian incentive-compatible mechanisms. In Appendix \ref{sec-dominant-strategy-lb}, we show that no two-sided dominant strategy incentive compatible mechanism provides a bounded approximation ratio. 

\section{The Setting} \label{sec-prelim}

In the bilateral trade problem, we have two agents: the seller and the buyer. The seller owns an indivisible item and his value for it is $s$. The buyer's value for the item is $b$. The values $b$ and $s$ are drawn from a joint distribution $\mathcal{F}$. 

A (direct) deterministic \emph{mechanism} $M$ for the bilateral trade problem consists of two functions $M=(x,p)$. For every tuple of seller and buyer values $(s,b)$ in the support of $\mathcal F$, $x(s,b)=1$ if a trade occurs and $x(s,b)=0$ otherwise. If there is a trade, $p(s,b)$ specifies the price that the buyer pays for the item and the payment that the seller gets for it. We require that $b\geq p(s,b)\geq s$. These restrictions on the payment are called \emph{ex-post individual rationality}.


The \emph{optimal welfare} of a joint distribution $\mathcal{F}$ is $\welfare{OPT}{\mathcal{F}} =\mathop{\mathbb{E}}_{(s,b)\sim \mathcal{F}} [ \max\{b,s\}]$.
The \emph{welfare} of a mechanism $M = (x,p)$ is $\welfare{M}{\mathcal{F}} =\mathop{\mathbb{E}}_{(s,b)\sim \mathcal{F}}[(x(s,b)\cdot b + (1-x(s,b))\cdot s]$.

For a joint distribution $\mathcal{F}$, the \emph{approximation ratio of a mechanism $M=(x,p)$ to the optimal welfare} is $\frac{\welfare{OPT}{\mathcal{F}}}{\welfare{M}{F}}$.



For a joint distribution $\mathcal F$, we denote by $\mathcal{F}_{s|b}$, the conditional cumulative distribution function of the seller, given that the buyer's value is $b$. Similarly, we denote by $\mathcal{F}_{b|s}$, the conditional cumulative distribution function of the buyer, given that the seller's value is $s$. We denote by $\indicator_A$ the indicator function for the event $A$. For example, we will use $\indicator_{[s>b]}$ to denote the indicator function for the event that the value of the seller $s$ is larger than the value of the buyer $b$.

In this paper, we consider several notions of incentive compatibility of a mechanism $M=(x,p)$. We start by defining incentive compatibility for only one of the players:
\begin{itemize}
    \item \textbf{Seller Dominant Strategy Incentive Compatibility:} for every $s,s',b$:
    $$x(s,b)\cdot p(s,b) + (1-x(s,b))\cdot s \geq x(s',b)\cdot p(s',b) + (1-x(s',b))\cdot s.$$
    \item \label{seller-ic} \textbf{Seller Bayesian Incentive Compatibility:}
    for every $s,s',b$: $$\mathop{\mathbb{E}}_{b\sim \mathcal{F}_{b| s} }[ x(s, b)\cdot p(s, b) + (1-x(s, b))\cdot s] \geq \mathop{\mathbb{E}}_{b\sim  \mathcal{F}_{b|s}}[ x(s', b)\cdot p(s', b) + (1-x(s', b))\cdot s].$$
    \item \textbf{Buyer Dominant Strategy Incentive Compatibility:} 
    for every $b,b',s$:
    $$x(s,b)\cdot (b- p(s,b)) \geq x(s,b')\cdot (b- p(s,b')).$$
    \item \label{buyer-ic} \textbf{Buyer Bayesian Incentive Compatibility:}
    for every $b,b',s$: $$\mathop{\mathbb{E}}_{s\sim \mathcal{F}_{s| b} }[ x(s,b)\cdot (b-p(s,b))]\geq \mathop{\mathbb{E}}_{s\sim \mathcal{F}_{s| b} }[ x(s,b')\cdot (b-p(s,b'))].$$
    
\end{itemize}
We will say that a mechanism is \emph{Dominant Strategy (Bayesian) Incentive Compatible} if it is dominant strategy (Bayesian) incentive compatible for both the buyer and the seller. A mechanism is \emph{one-sided dominant strategy incentive compatible} if it is dominant strategy incentive compatible for at least one of the players and Bayesian incentive compatible for the other. 





A distribution $\mathcal{F}_{b|s}$ is \emph{equal revenue distribution} for a seller with value $s$, if, for every value of $p$ in the support of the buyer's distribution $\mathcal{F}_{b|s}$, the expected payment to the seller from a take-it-or-leave-it offer of price $p$ to the buyer is the same.

For example, consider a seller with a value $0$, and a buyer with distribution $\mathcal{F}_{b|0}(b)$:
$$\mathcal{F}_{b|0}(b) = \begin{cases}
        1- \frac{1}{b} &  b \geq 1 ;\\
        0  & b < 1.
        \end{cases}
$$
The expected payment to the seller from any take-it-or-leave-it offer of price $1 \geq p$ to the buyer is $1$.

The expected profit of a buyer with value $b$, from a take-it-or-leave-it offer of price $p$ to the seller is $\mathcal{F}_{s|b}(p))\cdot(b-p)$. 

A distribution $\mathcal{F}_{s|b}$ is \emph{equal profit distribution} for a buyer with value $b$, if, for every value of $p$ in the support of the seller's distribution $\mathcal{F}_{s|b}$, the expected profit of the buyer from a take-it-or-leave-it offer of price $p$ to the seller is the same. For example, consider a buyer with value $1$, and a seller with distribution $\mathcal{F}_{s|1}(s)$:
$$\mathcal{F}_{s|1}(s) = \begin{cases}
         \frac{1}{e(1-s)} &  0 \leq s \leq 1-\frac{1}{e}; \\
        1  & s > 1-\frac{1}{e}.
        \end{cases}
$$
The expected profit of the buyer from any take-it-or-leave-it offer of price $ p \in [0, 1-\frac{1}{e}]$ to the seller is $\frac{1}{e}$. 

In the \emph{buyer-offering mechanism}, a buyer with value $b$ makes a take-it-or-leave-it offer of price $p$ to the seller, where the price $p$ maximizes the buyer's profit, i.e., maximizes $(b-p)\cdot \mathcal{F}_{s|b}(p)$.

\section{An $\frac e {e-1}$-Approximation for Correlated Values}\label{sec-algorithm}

In this section, we prove that the buyer-offering mechanism provides an $\frac e {e-1}\approx 1.58$ approximation to the optimal welfare, even if the values are drawn from a joint distribution. Recall that even when the values are drawn from independent distributions, the best currently known approximation mechanisms achieve a close approximation ratio of $\approx 1.38$ \cite{CW-STOC, CSTOC}. Our approximation guarantee holds for all distributions for which the buyer-offering mechanism is well defined (i.e., there always exists an offer that maximizes the profit), but note that there are distributions for which the buyer-offering mechanism is not well defined. For example, consider the following joint distribution $\mathcal{F}^\varepsilon$ for some small $0 < \varepsilon < 1$. In $\mathcal{F}^\varepsilon$, the buyer has only one value $1$, and the seller's value is supported on the interval $(0,\frac{e - \varepsilon -1 }{e-\varepsilon}]$, with marginal distribution function $\mathcal{F}^\varepsilon_s$:
$$
\mathcal{F}^\varepsilon_s(s) = \begin{cases}
    \frac{1+\varepsilon(1-s)}{e(1-s)} & s \in (0, \frac{e - \varepsilon -1 }{e-\varepsilon}];\\
    1 & s\geq \frac{e - \varepsilon -1 }{e-\varepsilon}.
\end{cases}
$$
Let $f(p) = (1-p)\cdot {F}^\varepsilon_s(p)$ be a function that denotes the expected profit of the buyer from a take-it-or-leave-it offer of $p$ to the seller. Then, $f(0)=0$, and for every $p \in (0, \frac{e - \varepsilon -1 }{e-\varepsilon}]$, we get $f(p) = \frac{1+\varepsilon(1-p)}{e}$. Observe that the derivative of $f$ for values $p \in (0, \frac{e - \varepsilon -1 }{e-\varepsilon}]$ is $\frac{-\varepsilon}{e}$, which is negative. Hence, the buyer's expected profit from a take-it-or-leave-it offer of $p$ to the seller is a strictly decreasing function in the interval $p \in (0, \frac{e - \varepsilon -1 }{e-\varepsilon}]$. Since the interval is open at 0, the function does not have a maximum. Thus, the buyer-offering mechanismis not defined for this distribution. 

We prove that the buyer-offering mechanism provides an $\frac e {e-1}$-approximation for all distributions for which it is defined. For the remaining distributions, we show that a slight variant of the buyer-offering mechanism provides a similar approximation ratio:

\begin{theorem} \label{bp-cons-app}\ \ 

\begin{enumerate}
        \item For every joint distribution $\mathcal F$ of the values of the buyer and seller, the buyer-offering mechanism provides an $\frac{e}{e-1}$ approximation to the optimal welfare if the buyer-offering mechanism is well-defined.
        \item For every joint distribution $\mathcal F$ of the values of the buyer and seller and every $\varepsilon>0$, there is a one-sided dominant strategy mechanism that provides an $\frac{e}{e-1}+O(\varepsilon)$ approximation to the optimal welfare.
\end{enumerate}
\end{theorem}

We first prove the first part of the theorem. We then use the first part to prove the second part. After establishing the theorem, we discuss extending the result to double auctions (Subsection \ref{subsec-double-auctions}).

\subsection{Proof of Theorem~\ref{bp-cons-app}: Part I}
    Fix a value $b$ of the buyer in the support of $\mathcal F$. Denote by $\mathcal{F}_{s|b}$ the distribution of the seller's value given that the value of the buyer is $b$. We will show that, for every $b$, the approximation ratio of the buyer-offering mechanism when the value of the buyer is always $b$ and the value of the seller is always $\mathcal{F}_{s|b}$ is $\frac e {e-1}$. This immediately implies that the approximation ratio of the buyer-offering mechanism when the values are drawn from $\mathcal F$ is $\frac e {e-1}$.


    Let $p_b$ be the price that the buyer offers when his value is $b$. Since the value $b$ is fixed, to simplify notation, we drop the subscripts from $p_b$ and $\mathcal{F}_{s|b}$, and simply write $p$ and $\mathcal F$ instead. 
     
     We now bound from above the approximation ratio of the buyer-offering for the distribution $\mathcal F$ (i.e., the expected optimal welfare divided by the expected welfare of the buyer-offering mechanism). 
     \begin{align}\label{app-first}
           \frac{b\cdot \mathcal{F}(b) + \mathbb{E}_{s \sim \mathcal{F}}[\indicator_{[s>b]}\cdot s]}{b\cdot \mathcal{F}(p) + \mathbb{E}_{s \sim \mathcal{F}}[\indicator_{[ p < s \leq b]}\cdot s]  + \mathbb{E}_{s \sim \mathcal{F}}[\indicator_{[s>b]}\cdot s]] } \leq 
           \frac{b\cdot \mathcal{F}(b)}{b\cdot \mathcal{F}(p) + \mathbb{E}_{s \sim \mathcal{F}}[\indicator_{[ p < s \leq b]}\cdot s]}.
     \end{align}
     Let $q_1= \mathcal{F}(p)$ and let $q_2 = \mathcal{F}(b)$. We have that $q_1 \leq q_2$ (since $p\leq b$). Observe that if $q_2 =q_1$, the approximation ratio is $1$, as needed. Therefore, we assume that $q_1 < q_2$. Rewriting the RHS of \eqref{app-first} we have:
     \begin{align}\label{app-second}
          \frac{b\cdot \mathcal{F}(b)}{b\cdot \mathcal{F}(p) + \mathbb{E}_{s \sim \mathcal{F}}[\indicator_{[ p < s \leq b]}\cdot s]}= \frac{b\cdot q_2}{b\cdot q_1 + \mathbb{E}_{s \sim \mathcal{F}}[s| p < s \leq b]\cdot (q_2-q_1)}.
     \end{align}    
     For fixed $q_1, q_2$ and $b$, this expression is maximized when $\mathbb{E}_{s \sim \mathcal{F}}[s| p < s \leq b]\cdot (q_2-q_1)$ is minimized. Therefore, in Lemma~\ref{minimal-expectation}, we bound from below the expression $\mathbb{E}_{s\sim \mathcal{F}}[s| p < s \leq b]\cdot (q_2-q_1)$ to achieve an upper bound on the approximation ratio.
     \begin{lemma} \label{minimal-expectation}
         $\mathbb{E}_{s \sim \mathcal{F}}[s| p < s \leq b]\cdot (q_2-q_1) \geq q_1 \cdot (b-p)\cdot \ln{\frac{q_1}{q_2}} +b\cdot (q_2-q_1).$
     \end{lemma}
    \begin{proof}
   
        Recall that the price $p$ maximizes the buyer's profit for the distribution $\mathcal{F}$. We use this to bound from above $\mathcal{F}(p')$ for $p < p' < b$. Let $u = q_1\cdot (b-p)$ be the expected profit of the buyer when the price is $p$. Then:
        \begin{align}\label{bp-proof-cdf-bound}
             u  \geq \mathcal{F}(p')(b-p')  
            \iff 
            \mathcal{F}(p') & \leq \frac{u}{(b-p')}.
        \end{align}
        For $p' = b-\frac{u}{q_2}$, we get that the bound \eqref{bp-proof-cdf-bound} is tight and equal to $q_2$\footnote{Observe that $b-\frac{u}{q_2} \geq p$ since $ b-\frac{q_1(b-p)}{q_2} \underbrace{\geq}_{q_1\leq q_2} p$.}.
        Let $f$ be the seller's probability density function given that the buyer's value is $b$.
        We want to bound from below the expression:
        \begin{align*}
            \mathbb E_{s\sim\mathcal F}[\indicator_{[ p < s \leq b]}\cdot s] & = \int_{p}^b sf(s) \d s \underbrace{=}_{\text{integration by parts}} \eval{s\mathcal F(s)}_p^b -\int_p^b \mathcal F(s) \d s \geq b\cdot q_2 -p\cdot q_1 -\int_p^b \mathcal F(s) \d s.
        \end{align*}
        We use the following bounds on $\mathcal F(s)$: for $p \leq s \leq b -\frac{u}{q_2}$ we have $\mathcal{F}(s) \leq \frac{u}{b-s}$ (by Inequality~\eqref{bp-proof-cdf-bound}) and for $b -\frac{u}{q_2} \leq s \leq b$ we have $\mathcal{F}(s) \leq q_2$. Now, 
        \begin{align*}
            \mathbb E_{s\sim \mathcal F}[\indicator_{[ p < s \leq b]}\cdot s] & \geq b\cdot q_2 -p\cdot q_1 -\int_p^{b-\frac{u}{q_2}} \frac{u}{b-s} \d s -\int_{b-\frac{u}{q_2}}^{b} q_2 \d s =  b\cdot q_2 -p\cdot q_1 -u +\eval{u\cdot \lnn{b-s}}_{p}^{b-\frac{u}{q_2}}\\
            & = b\cdot q_2 -p\cdot q_1 -u +u\cdot \lnn{\frac{\frac{u}{q_2}}{b-p}} = b\cdot q_2 -p\cdot q_1 -(q_1(b-p)) +q_1(b-p)\lnn{\frac{q_1}{q_2}} \\
            & = q_1(b-p)\lnn{\frac{q_1}{q_2}} + b\cdot(q_2 - q_1).
        \end{align*}
    \end{proof}

    Overall, we get that the approximation ratio is bounded from above by
    \begin{equation*}
        \frac{b\cdot q_2}{b\cdot q_1 + q_1\cdot (b-p)\cdot \ln{\frac{q_1}{q_2}} + b\cdot (q_2-q_1)} \underbrace{\leq}_{\substack{p \geq 0 ,\   \ln{\frac{q_1}{q_2}} \leq 0}} \frac{b\cdot q_2}{ b\cdot q_1 \cdot \ln{\frac{q_1}{q_2}}+b\cdot q_2} = \frac{1}{ \frac{q_1}{q_2} \cdot \ln{\frac{q_1}{q_2}} + 1} \underbrace{\leq}_{\footnotemark{}} \frac{e}{e-1}.
    \end{equation*}

\footnotetext{Recall that $q_1 \leq q_2$, and so $0 \leq \frac{q_1}{q_2} \leq 1$. Thus, the function $f(x)=\frac{1}{1+x\ln x}$ is maximized when $x=\frac{1}{e}$, and its maximal value is $\frac{e}{e-1}$.} 

\subsection{Proof of Theorem~\ref{bp-cons-app}: Part II}

We consider a slight variant of the buyer-offering mechanism. Given $\varepsilon>0$,  set $\delta=\varepsilon\cdot OPT$, where $OPT$ is the expected value of the optimal social welfare. The $\varepsilon$-buyer offering mechanism makes a profit-maximizing take-it-or-leave-it offer $p$ for the seller, where $p$ belongs to the set of offers that are a multiple of $\delta$. I.e., $p=k\cdot \delta$ for some $k \in \mathbb N$. This mechanism is Bayesian incentive compatible since the set of possible prices for a buyer with value $b$ contains at most  $ \lceil \frac{b}{k} \rceil$ elements. Since this set is finite for every value $b$, it has a maximum-profit element. In addition, the seller obviously has a dominant strategy. 

As in the first part, we prove the approximation guarantee for every value $b$ of the buyer. Let ${\distributionlbbp}_s(s)$ be the marginal distribution of the seller given $b$. Obtain a ``discretized'' ${\distributionlbbp}'_s(s)$ by ``pushing'' the mass of ${\distributionlbbp}_s(s)$ in all points that are not a multiple of $\delta$ to the nearest (from above) multiple of $\delta$. Note that the buyer offering mechanism is now defined since the support of ${\distributionlbbp}'_s(s)$ is finite. Thus, there is an offer with a maximum profit since for every offer that is not in the support there is an offer in the support with at least the same profit. Furthermore, observe that for every $b$, if the buyer offering mechanism makes an offer $p$ when the marginal distribution is ${\distributionlbbp}'_s(s)$ then $p$ is also the offer that $\varepsilon$-buyer offering mechanism makes when the marginal distribution is ${\distributionlbbp}_s(s)$. Observe that ${\distributionlbbp}_s(s)$ and ${\distributionlbbp}'_s(s)$ are very close to each other, and thus the expected welfare that the buyer offering mechanism provides for ${\distributionlbbp}'_s(s)$ and the expected welfare of the $\varepsilon$-buyer offering mechanism for ${\distributionlbbp}_s(s)$ differ only by $\delta$. The second part of the theorem now follows since, by the first part, the buyer-offering mechanism provides an $\frac e {e-1}$ approximation for ${\distributionlbbp}'_s(s)$ and because the optimal welfare of ${\distributionlbbp}'_s(s)$ and the optimal welfare of ${\distributionlbbp}_s(s)$ differ only by $\delta$.

\subsection{Tightness of Analysis}

We now present an instance of a distribution $\distributionlbbp$ where the buyer-offering mechanism yields an approximation ratio no better than $\frac{e}{e-1}$ to the optimal welfare. Subsequently, in Subsection \ref{one-sided-lb}, we establish a more robust result that states that no one-sided dominant strategy mechanism can offer an approximation ratio better than $\frac e {e-1}$. Since the buyer-offering mechanism is a one-sided dominant strategy mechanism, the result of Subsection \ref{one-sided-lb} also implies that the analysis of the buyer-offering mechanism is precise. Nevertheless, we provide a direct analysis of the buyer-offering mechanism in this section for simplicity.

Let $\distributionlbbp$ be a joint distribution over the buyer and seller values in which the value of the buyer is always $1$, and the value of the seller is in $[0, \frac{e-1}{e}]$. The seller's value is distributed as follows:
$${\distributionlbbp}_s(s) = \begin{cases}
        \frac{\frac{1}{e}}{1-s} & s \in [0, \frac{e-1}{e}];\\
        \, \, 1 &  s>\frac{e-1}{e}.
        \end{cases}
$$
Note that the seller's distribution is an equal profit distribution for a buyer with value of $1$. I.e., for the buyer, every price $p \in [0, \frac{e-1}{e}]$ yields the same expected profit. Thus, by tie-breaking, we may assume that the buyer-offering mechanism uses price $0$ (alternatively, to eliminate the use of tie-breaking, one can change ${\distributionlbbp}_s$ and slightly increase the probability that the seller's value is $0$ -- the analysis remains almost identical). 

We now analyze the approximation ratio of the buyer-offering mechanism for the distribution $\distributionlbbp$:

\begin{align*}
    \frac{\welfare{OPT}{\distributionlbbp}}{\welfare{ALG}{\distributionlbbp}} &=  \frac{b\cdot\Pr_{{\distributionlbbp}_s}(s\leq p) + b\cdot\Pr_{{\distributionlbbp}_s}(p < s\leq b)}{b\cdot\Pr_{{\distributionlbbp}_s}(s\leq p) + \mathbb{E}_{s \sim {\distributionlbbp}_s}[s| p < s \leq b]\cdot\Pr_{{\distributionlbbp}_s}(p < s\leq b)}\\
    & = \frac{1}{\frac{1}{e} + \mathbb{E}_{s \sim {\distributionlbbp}_s}[s| 0 < s \leq 1] \cdot (1-\frac{1}{e})}\\
    & = \frac{1}{\frac{1}{e} + \frac{1}{e}\cdot \ln{\frac{1}{e}} + 1-\frac{1}{e}} = \frac{e}{e-1}.
\end{align*}

\subsection{Double Auctions}\label{subsec-double-auctions}

In double auctions there are multiple buyers and sellers, each seller $i$ owns a single item and his value for it is $s_i$, all items are identical, and each buyer $j$ wants one unit and his value for it is $b_j$. Similarly, to bilateral trade, we wish to approximate the optimal welfare, which in the double auction case with $k$ sellers is equal to the sum of the $k$ largest values among the sellers and buyers. 

Our positive result for bilateral trade also implies a positive result for double auctions. 
Similarly to the work of Babaioff et al. \cite{BothWorlds}, to obtain a solution for the double auction case, we combine McAfee's trade reduction mechanism \cite{TradeReduction} with the mechanism we use for the bilateral trade case (the buyer offering) in the following manner: compute the maximal number of efficient trades (where trade is efficient only if the value of the buyer is larger than the value of the seller). Run the trade reduction mechanism if there are at least two efficient trades. If there is only one efficient trade, run the buyer offering mechanism for the bilateral trade problem with the buyer being the highest-value buyer and the seller being the lowest-value seller. The distribution over the seller's value is the conditional distribution over the value of this seller given all values except his own and that he is the lowest value seller and the additional requirement that the price must be at least as large as the value of the second highest buyer.

Similarly to \cite{BothWorlds}, this mechanism is Bayesian incentive compatible and ex-post individually rational. Observe that the approximation ratio of this mechanism is at least $2$. If there are at least two efficient trades, the approximation ratio of the trade reduction mechanism is $1-\frac{1}{k}$ where $k$ is the number of efficient trades, which gives us a $2$ approximation guarantee. If there is only one efficient trade, we get an approximation guarantee that is at least as good as the approximation guarantee of the buyer-offering mechanism for the bilateral trade case, which is $\frac{e}{e-1}$. Lastly, if the number of efficient trades is $0$, we do nothing and get an optimal approximation.


\section{The Limits of One-Sided Dominant Strategy Mechanisms}\label{one-sided-lb}


We now prove that no one-sided dominant strategy mechanism can provide an approximation ratio better than $\frac e {e-1}$. This shows, in particular, the optimality of the buyer-offering mechanism as its approximation ratio is $\frac e {e-1}$ (Theorem \ref{bp-cons-app}). It also shows, for example, that taking the better of the buyer-pricing mechanism and a fixed price mechanism does not improve the approximation ratio. In Subsection~\ref{subsec-one-sided-distribution}, we present a specific distribution for which every one-sided dominant strategy mechanism does not provide an approximation ratio better than $\frac e {e-1}$. We provide the formal analysis of the impossibility in Subsection~\ref{subsec-one-sided-lower-bound}.

\subsection{The Distribution $\distributionlbonesided$}\label{subsec-one-sided-distribution}

For $k\in \mathbb N, \varepsilon>0$, let $\distributionlbonesided$ be a joint distribution over the buyer and seller values in which the buyer receives values in $[1,k] \cup \{ k+1 \}$ with probability density function of $\frac{1}{(b+\varepsilon)^{2}}$ for $1 \leq b \leq k$ and of $\frac{1}{k+\varepsilon}+\frac{\varepsilon}{1+\varepsilon}$ for $b= k+1$. 
For every value $b$ in $[1,k] \cup \{ k+1 \}$, the seller's value is distributed according to an (almost) equal profit distribution of the buyer. I.e., when the buyer's value is $b$, the cumulative distribution function of the value of the seller is $\distributionlbonesidedcdfs$:

$$\distributionlbonesidedcdfs(s) = \begin{cases}
        \frac{\frac{b}{e}}{b-s+\varepsilon} & 0\leq s \leq  b\cdot \frac{e-1}{e} + \varepsilon;\\
        1 &  s>b\cdot \frac{e-1}{e} + \varepsilon.
        \end{cases}
$$
Note that the buyer's marginal distribution always sums up to 1:
\begin{multline*}
\distributionlbonesided_b(k+1)= \mathrm{\Pr}_{b \sim \distributionlbonesided_b}(b=k+1) + \int_1^k \frac{1}{(b+\varepsilon)^2}\d b = \frac{1}{k+\varepsilon}+\frac{\varepsilon}{1+\varepsilon} + \eval{-\frac{1}{b+\varepsilon}}_{1}^{k} \\
= \frac{1}{k+\varepsilon}+\frac{\varepsilon}{1+\varepsilon} - \frac{1}{k+\varepsilon}+\frac{1}{1+\varepsilon} = 1. 
\end{multline*}
Moreover, for every value in the buyer's support, the conditional distribution of the seller sums up to 1:  $\distributionlbonesidedcdfs(b\cdot \frac{e-1}{e}+ \varepsilon)=   \frac{\frac{b}{e}}{b-(b\cdot \frac{e-1}{e}+ \varepsilon) +\varepsilon} = 1$.

This distribution has two useful properties. The first is that when the seller's value is $0$, the buyer's distribution is very close to an equal revenue distribution (it is implied by the marginal probability density function of the buyer that is close to $\frac{1}{b^2}$). The second property is that for every value $b$ in the buyer's support, the seller's distribution is very close to an equal profit distribution.

Fix some mechanism $M$ that is one-sided dominant strategy for one of the players. If $M$ is dominant strategy for the buyer, our bound on the welfare is achieved by utilizing the first property. In this case, only the seller's value can affect the value of the offer. Given that the seller's value is $0$, the buyer's distribution is very close to an equal revenue distribution. We show that when the seller's value is $0$, the seller will strictly prefer higher prices. Since the mechanism is Bayesian incentive compatible for the seller, the highest take-it-or-leave it offer is when $s=0$. Very roughly speaking, this implies that the welfare of the mechanism is low: if the highest offer is low, then the contribution to the welfare of trades when $s=0$ is high, but no trade is done for larger values of the seller, which happens with significant probability. On the other hand, if the value of the highest offer $p$ is large, trade is less likely in the $s=0$ case. 

If $M$ is dominant strategy for the seller, we utilize the second property. Now, only the buyer's value can affect the offer price. However, the buyer strictly prefers lower offers as lower prices will yield higher profit. Since the mechanism is incentive compatible for the buyer, the offer price is the same for every value $b$ of the buyer (otherwise, the buyer will prefer the lower offer and deviate from his equilibrium strategy). Then, every mechanism for $\distributionlbonesided$ that offers a take-it-or-leave-it-offer $p$ to the seller (one-sided dominant strategy for the buyer) is a fixed price mechanism. It only remains to show that a fixed price mechanism has low welfare.    


\subsection{Analysis of One-Sided Mechanisms for the Distribution $\distributionlbonesided$}\label{subsec-one-sided-lower-bound}
\begin{theorem}\label{one-sided-lower-bound} 
 Let $k\geq 2$ and $\varepsilon >0 $. Every one-sided dominant strategy mechanism for $\distributionlbonesided$ provides 
an approximation ratio of at least $\frac{e}{e-1}$ as $\varepsilon$ approaches $0$ and $k$ approaches $\infty$.
\end{theorem}

To prove this theorem we use the family of allocation rules $x_p$ (Definition~\ref{x-p-one-sided}). We show that the welfare of every one-sided dominant strategy mechanism for the seller is at most the welfare of $x_p$, for some $p$ (Claim~\ref{lb-one-sided-seller}). Similarly, we show that the welfare of every one-sided dominant strategy mechanism for the buyer is no better than the welfare of $x_p$, for some $p$ (Claim~\ref{lb-one-sided-buyer}). We conclude the proof of the theorem by bounding the approximation ratio of every possible allocation rule $x_p$ (Lemma~\ref{one-sided-lb-max-price-at-s-0}).

\begin{definition}\label{x-p-one-sided}
For every $p\geq 0$, let $x_p$ be the following allocation rule for the distribution $\distributionlbonesided$:
$$
x_p(b,s) = 
\begin{cases}
        1 & b=k+1; \\
        1  & s = 0 \text{ and }  k \geq b \geq p;\\
        0  & s = 0 \text{ and } b < p \text{ and } b \neq k+1;\\
        1  & s \neq  0 \text{ and } s \leq p;\\
        0 &  \text{otherwise}.
        \end{cases}
$$
\end{definition}
Observe that for some values of $p$, this allocation rule $x_p$, is not implementable. It is used simply to bound the welfare of every one-sided dominant strategy mechanism. The next three claims suffice to prove the theorem. Their proofs can be found in Subsections \ref{subsec-lb-one-sided-buyer}, \ref{subsec-lb-one-sided-seller}, and \ref{subsec-lb-max-price-at-s-0}.

\begin{claim} \label{lb-one-sided-buyer}
    Let $k\geq 2$ and $\varepsilon >0 $. 
    Let $M$ be a one-sided dominant strategy mechanism for the buyer. There exists $p \geq 0$, such that the welfare of $M$ is at most the welfare of $x_p$, both with respect to the distribution $\distributionlbonesided$.    
\end{claim}

\begin{claim} \label{lb-one-sided-seller}
    Let $k\geq 2$ and $\varepsilon >0 $. 
    Let $M$ be a one-sided dominant strategy mechanism $M$ for the seller. There exists $p \geq 0$, such that the welfare of $M$ is at most the welfare of $x_p$, both with respect to the distribution $\distributionlbonesided$.     
\end{claim}

 \begin{claim} \label{one-sided-lb-max-price-at-s-0}
    Fix $p\geq 0$. When $\varepsilon$ approaches $0$ and $k$ approaches $\infty$, the allocation rule $x_p$ provides 
    an approximation ratio no better than $\frac{e}{e-1}$ to the optimal welfare for the distribution $\distributionlbonesided$.
\end{claim}

\subsubsection{Mechanisms with Dominant Strategy for the Buyer (Proof of Claim~\ref{lb-one-sided-buyer})}\label{subsec-lb-one-sided-buyer}
    Every one-sided dominant strategy mechanism for the buyer is a take-it-or-leave-it offer for the buyer, where the price may only depend on the seller's value. 
    Fix such a mechanism $M=(x,p)$ for the distribution $\distributionlbonesided$. For every value $s$ in the seller's support, denote the take-it-or-leave-it offer of the mechanism $M$ by $p_s$.
    Now, by Lemma~\ref{lemma-seller-prefers-higher-prices}, the price $p_0$ offered to the buyer when the seller's value is $0$, should be no lower than any price in $\{p_s | s \in [0, (k+1)\cdot\frac{e-1}{e} +\varepsilon]\}$. This is true since if there is a price $p_s' > p_0$, and the seller prefers higher prices for $p_0$ (Lemma~\ref{lemma-seller-prefers-higher-prices}), the seller will play the strategy that sets the price to $p_s'$, in contradiction to the incentive compatibility of the mechanism.  
    Recall that the item can be traded only if the seller's value is at most the price. Thus, the mechanism $M$ can only trade when $s<p_s\leq p_0$.
    Now, let $p=p_0$, and consider the allocation rule $x_p$ (Definition~\ref{x-p-one-sided}). Observe that the allocation rule $x_p$ trades the item in every instance that $M$ trades the item and might trade the item when $M$ does not. Therefore, the welfare of $x_p$ is at least the welfare of $M$.
    \begin{lemma} \label{lemma-seller-prefers-higher-prices}
    For $s=0$ and $p \in [1,k]$, the expected profit of the seller from a take-it-or-leave-it offer with price $p$ to the buyer is higher than a take-it-or-leave-it offer with price $p' < p$. 
    
    \end{lemma}
    \begin{proof}[Proof of Lemma~\ref{lemma-seller-prefers-higher-prices}]
    Recall that the conditional cumulative distribution function of the buyer, given that the seller's value is $0$ is denoted by  $\distributionlbonesided_{b|0}$. 
    The profit of the seller with value $0$ from a take-it-or-leave-it offer of $p \in [1,k] \cup \{k+1\}$ is $(1-\distributionlbonesided_{b|0}(p))\cdot p$.
    To analyze this expression, we first provide an explicit expression for $1-\distributionlbonesided_{b|0}(p)$, where $p \in [1,k]$.
    By definition, we have:
    \begin{align*}
    1- \distributionlbonesided_{b|0}(p) &= \left(\int_p^{k} \dfrac{f^{k, \varepsilon} (0, b)}{ \mathrm{\Pr}_{s \sim \distributionlbonesided_s}(s=0)} \d b \right) + \frac{\mathrm{\Pr}_{\distributionlbonesided}(0, k+1)}{\mathrm{\Pr}_{s \sim \distributionlbonesided_s}(s=0)}\\
    & =\frac{1}{e\cdot \mathrm{\Pr}_{s \sim \distributionlbonesided_s}(s=0)} \left(\frac{k+1}{k+1+\varepsilon}\left(\frac{1}{k+\varepsilon}+\frac{\varepsilon}{1+\varepsilon}\right) +\int_p^{k} \frac{b}{(b+\varepsilon)^3} \d b \right)\\
    & =\frac{1}{e\cdot \mathrm{\Pr}_{s \sim \distributionlbonesided_s}(s=0)} \left(\frac{k+1}{k+1+\varepsilon}\left(\frac{1}{k+\varepsilon}+\frac{\varepsilon}{1+\varepsilon}\right) +\eval{ -\frac{\varepsilon +2b}{2(b+\varepsilon)^2}}_p^{k}\right)\\
    & = \frac{1}{e\cdot \mathrm{\Pr}_{s \sim \distributionlbonesided_s}(s=0)} \left(\frac{k+1}{k+1+\varepsilon}\left(\frac{1}{k+\varepsilon}+\frac{\varepsilon}{1+\varepsilon}\right) + \left(-\frac{\varepsilon +2k}{2(k+\varepsilon)^2}+ \frac{\varepsilon +2p}{2(p+\varepsilon)^2}\right)\right).
    \end{align*}
    Let $g(p) = p \left(\frac{k+1}{k+1+\varepsilon}\left(\frac{1}{k+\varepsilon}+\frac{\varepsilon}{1+\varepsilon}\right) + \left(-\frac{\varepsilon +2k}{2(k+\varepsilon)^2}+ \frac{\varepsilon +2p}{2(p+\varepsilon)^2}\right)\right)$. We now prove that $g$ is strictly increasing by showing that its derivative is positive. We get that the profit of the seller with value $0$ from a take-it-or-leave-it offer of  $p \in [1,k]$ is strictly increasing in $p$, as $g(p) = {p\cdot \left(1- \distributionlbonesided_{b|s=0}(b)\right)}\cdot {e\cdot \mathrm{\Pr}_{s \sim \distributionlbonesided_s}(s=0)}$.

    \begin{align*}
        g'(p) & = \frac{k+1}{k+1+\varepsilon}\left(\frac{1}{k+\varepsilon}+\frac{\varepsilon}{1+\varepsilon}\right) + \left(-\frac{\varepsilon +2k}{2(k+\varepsilon)^2}+ \frac{\varepsilon +2p}{2(p+\varepsilon)^2}\right) -p\frac{p}{(\varepsilon+p)^3}\\
        & = \frac{2(k+1)(k+\varepsilon)-(\varepsilon +2k)(k+1+\varepsilon)}{2(k+1+\varepsilon)(k+\varepsilon)^2} +\frac{\varepsilon(k+1)}{(1+\varepsilon)(k+1+\varepsilon)} + \frac{\varepsilon +2p}{2(p+\varepsilon)^2} -\frac{p^2}{(\varepsilon+p)^3}\\
        & = \frac{-k\varepsilon +\varepsilon - \varepsilon^2}{2(k+1+\varepsilon)(k+\varepsilon)^2} +\frac{\varepsilon(k+1)}{(1+\varepsilon)(k+1+\varepsilon)} + \frac{\varepsilon +2p}{2(p+\varepsilon)^2} -\frac{p^2}{(\varepsilon+p)^3}\\
        & = \frac{(-k\varepsilon +\varepsilon - \varepsilon^2)(1+\varepsilon)+2\varepsilon(k+1)(k+\varepsilon)^2}{2(k+1+\varepsilon)(k+\varepsilon)^2(1+\varepsilon)} + \frac{\varepsilon +2p}{2(p+\varepsilon)^2} -\frac{p^2}{(\varepsilon+p)^3}\\        
        & > \frac{(1+\varepsilon)(\varepsilon - \varepsilon^2) +k\varepsilon(2-1-\varepsilon)}{2(k+1+\varepsilon)(k+\varepsilon)^2(1+\varepsilon)} + \frac{\varepsilon +2p}{2(p+\varepsilon)^2} -\frac{p^2}{(\varepsilon+p)^3}\\
        & > \frac{\varepsilon +2p}{2(p+\varepsilon)^2} -\frac{p^2}{(\varepsilon+p)^3}.\\
    \end{align*}
    where in the last inequality we assume that $\varepsilon - \varepsilon^2 >0$, since we can choose $\varepsilon>0$ to be as small as we want. We prove that $\frac{\varepsilon +2p}{2(p+\varepsilon)^2} -\frac{p^2}{(\varepsilon+p)^3} >0$, which implies that $g'(p) > 0$ for every $p \in [1,k]$:

    \begin{align*}
        & \frac{\varepsilon +2p}{2(p+\varepsilon)^2} > \frac{p^2}{(\varepsilon+p)^3} \iff (\varepsilon +2p)(p+\varepsilon) > 2p^2 \iff 2p^2 +3p\varepsilon +\varepsilon^2 > 2p^2 \iff 3p\varepsilon +\varepsilon^2 >0.
    \end{align*}
    This proves that the expected profit of the seller with value $0$ from a take-it-or-leave-it offer of $p \in [1,k]$ is smaller than his expected profit from an offer of $p' \in [1,k]$ that is smaller than $p$. Observe that the expected profit of the seller with value $0$ from a take-it-or-leave-it offer of $p<1$ is even smaller than his expected profit from an offer of price $1$ (as reducing the offer below $1$, does not increase the probability that the buyer will buy the item). The profit of an offer $p=k+1$ is even larger profit than an offer of $k$, as the probability $1-\distributionlbonesided_{b|s=0}(k+1)$ is $1-\distributionlbonesided_{b|s=0}(k)$ but the price is larger ($k+1>k$). This concludes the proof of the lemma. 
\end{proof}
    
\subsubsection{Mechanisms with Dominant Strategy for the Seller (Proof of Claim~\ref{lb-one-sided-seller})}\label{subsec-lb-one-sided-seller}

    Every one-sided dominant strategy mechanism for the seller is a take-it-or-leave-it offer to the seller, where the offer depends only on the value of the buyer. 
    Fix a mechanism $M=(x,p)$ for the distribution ${\distributionlbonesided}$. For every $b \in [1,k]\cup \{ k+1 \}$, the mechanism offers a price $p_b$.
    By Lemma~\ref{lemma-buyer-prefers-lower-prices}, the expected profit of the buyer with value $b \in [1,k]\cup \{ k+1 \}$ is higher when $p_b$ is smaller. Since the mechanism is Bayesian incentive compatible for the buyer we claim that it must be a fixed price, i.e., $p_b$ is equal for every $b \in [1,k]\cup \{ k+1 \}$. This is true since if there are two different values $p_b > p_{b'}$ and the buyer strictly prefers lower prices (Lemma~\ref{lemma-buyer-prefers-lower-prices}), the buyer will play the strategy that sets the price to $p_{b'}$, in contradiction to the incentive compatibility of the mechanism.
    Let $p$ be the fixed price and consider the allocation rule $x_p$ (see Definition~\ref{x-p-one-sided}).
    Recall that the item can be traded only if the seller's value is at most the price. Thus, the mechanism $M$ only trades the item when $s\leq p$. Observe that the allocation rule $x_p$ trades the item in every instance that $M$ trades the item and might trade the item when $M$ does not. Thus, the welfare of $x_p$ is at least the welfare of $M$.

    \begin{lemma} \label{lemma-buyer-prefers-lower-prices}
    For $b \in [1,k] \cup \{ k+1 \}$ and $p \in [0,b\cdot \frac{e-1}{e}+ \varepsilon]$, the expected profit of the buyer from a take-it-or-leave-it offer with price $p$ to the seller is higher than a take-it-or-leave-it offer with price $p' > p$.  
\end{lemma}
    \begin{proof}[Proof of Lemma~\ref{lemma-buyer-prefers-lower-prices}]
    
    Let $b\in [1,k]\cup \{ k+1 \}$. The conditional cumulative distribution of the seller's value is $\distributionlbonesidedcdfs$. Now, for every $p \in  [0,b\cdot \frac{e-1}{e}+ \varepsilon]$, the expected profit of the buyer from a take-it-or-leave-it offer of price $p$ to the seller is $(b-p)\cdot \frac{b}{e(b-p+\varepsilon)}$. We define the function $g(p) = (b-p)\cdot \frac{b}{e(b-p+\varepsilon)}$ for every $p \in [0,b\cdot \frac{e-1}{e}+ \varepsilon]$, and show that it is a strictly decreasing function. By definition, the function $g(p)$ is the expected profit of the buyer from a take-it-or-leave-it offer of price $p$ to the seller. Intuitively, if $\distributionlbonesidedcdfs$ was exactly the equal profit distribution, i.e., $\distributionlbonesidedcdfs (p) = \frac{b}{e(b-p)}$ for every $p\in  [0,b\cdot \frac{e-1}{e}+ \varepsilon]$, then $g$ was a constant function with value $\frac{b}{e}$. However, since $\distributionlbonesidedcdfs$ is a slightly skewed equal profit distribution, lower prices yield strictly higher profits, and so $g$ is a strictly decreasing function. Formally, $g(p)$ is a decreasing function since its derivative is negative for every value $p$ in its range:
   
    $$
    g'(p) = \frac{b}{e}\left(\frac{-\left(b-p+\varepsilon\right)--\left(b-p\right)}{\left(b-p+\varepsilon\right)^2}\right)=\frac{-b\varepsilon}{e\left(b-p+\varepsilon\right)^2}\underbrace{<}_{b>0} 0.
    $$

\end{proof}

\subsubsection{Bounding the Approximation Ratio (Proof of Claim~\ref{one-sided-lb-max-price-at-s-0})}\label{subsec-lb-max-price-at-s-0}
We start with computing the optimal welfare of the distribution $\distributionlbonesided$. Observe that in every instance in the distribution $\distributionlbonesided$, the buyer has a higher value than the seller, thus trade occurs in the optimal allocation rule. The optimal welfare is therefore:
\begin{align*}
     \left(k+1\right)\left(\frac{1}{k+\varepsilon}+ \frac{\varepsilon}{1+ \varepsilon}\right) + \int_{1}^k \frac{1}{\left(b+\varepsilon\right)^2}\cdot b \d b & = \left(k+1\right)\left(\frac{1}{k+\varepsilon}+ \frac{\varepsilon}{1+ \varepsilon}\right) + \eval{\frac{\varepsilon}{\varepsilon + b} +\lnn{b+\varepsilon}}_1^k \\
    & = \left(k+1\right)\left(\frac{1}{k+\varepsilon}+ \frac{\varepsilon}{1+ \varepsilon}\right) + 
    \varepsilon\left(\frac{1}{k+\varepsilon}- \frac{1}{1+ \varepsilon}\right) + \lnn{\frac{k+\varepsilon}{1+\varepsilon}}\\
    & = \left(\frac{k+1+\varepsilon}{k+\varepsilon}+ \frac{k\varepsilon}{1+ \varepsilon}\right) + \lnn{\frac{k+\varepsilon}{1+\varepsilon}}.
\end{align*}
This expression approaches $\ln k + 1+\frac{1}{k}$ as $\varepsilon$ approaches $0$.

We consider the possible values of the buyer's value and compute the contribution of each value to the expected welfare when the item is traded according to $x_p$. When $b=k+1$, $x_p$ always sells the item (as does the optimal allocation rule). This  contributes $(k+1)\left(\frac{1}{k+\varepsilon}+ \frac{\varepsilon}{1+ \varepsilon}\right)$ to the expected welfare.

For every $b \in [1,k]$, with probability $\frac{b}{e(b+\varepsilon)}$, the seller's value is $0$. Then, according to $x_p$, the item is traded only when $b\geq p$. These instances contribute $\int_p^k f^{k, \varepsilon}_b(b) \cdot f_s^{k, \varepsilon}(0) \cdot b \d b$ to the expected welfare.
When  $b \in [1,k]$, and the seller's value is larger than $0$, the item is traded only when $s \leq p$ according to $x_p$. These instances contribute  

$\int_{1}^k f^{k, \varepsilon}_b(b) [(b\cdot (\distributionlbonesidedcdfs(p)- \distributionlbonesidedcdfs(0)) + \int_{p}^{b\cdot \frac{e-1}{e}} s\cdot \distributionlbonesidedpdfs(s)  \d{s})] \d{b}$ 
to the expected welfare.  Thus, the expected welfare is the sum of these three expressions:

\begin{align}
    & (k+1)\left(\frac{1}{k+\varepsilon}+ \frac{\varepsilon}{1+ \varepsilon}\right) \label{first-expression-x-p-lemma}\\
     & \int_p^k f^{k, \varepsilon}_b(b) \cdot f_s^{k, \varepsilon}(0) \cdot b \d b \label{second-expression-x-p-lemma}\\
     & \int_{1}^k f^{k, \varepsilon}_b(b) \left(b\cdot \left(\distributionlbonesidedcdfs(p)- \distributionlbonesidedcdfs(0)\right) + \int_{p}^{b\cdot \frac{e-1}{e}} s\cdot \distributionlbonesidedpdfs(s)  \d{s}\right) \d{b} \label{third-expression-x-p-lemma}
\end{align}

Next, we consider each expression separately, bound it from above, and take its limit as $\varepsilon$ goes to $0$. We start with the first expression~\eqref{first-expression-x-p-lemma}:

\begin{equation*}
    \lim_{\varepsilon \to 0^+} (k+1)\left(\frac{1}{k+\varepsilon}+ \frac{\varepsilon}{1+ \varepsilon}\right) = 1+\frac{1}{k}.
\end{equation*}

We continue with the second expression~\eqref{second-expression-x-p-lemma}. By the dominated convergence theorem, we can swap the order of the integral and the limit operator\footnote{\label{dominated-conv-thm-footnote}All the expressions we consider are bounded (for example, by expected welfare of the distribution, which is at most $\frac{2}{k+2} + \lnn{k+1}$). In addition, the sequence of functions $f_{\varepsilon}$ converges point-wise to the function $f_0$, for every $f$ that we consider.}.

\begin{align*}
     \lim_{\varepsilon \to 0^+} \int_p^k f^{k, \varepsilon}_b(b) \cdot f_s^{k, \varepsilon}(0) \cdot b \d b & =\lim_{\varepsilon \to 0^+} \int_p^k \frac{1}{(b+\varepsilon)^2}\cdot \frac{b}{e\cdot (b+\varepsilon)} \cdot b \d b = \lim_{\varepsilon \to 0^+} \int_p^k \frac{b^2}{e\cdot (b+\varepsilon)^3} \d b \\
     & = \int_p^k \frac{b^2}{eb^3} \d b = \frac{1}{e}\lnn{\frac{k}{p}}.\\
\end{align*}

Finally, for the third expression~\eqref{third-expression-x-p-lemma}, we first break it into three expressions:
\begin{align*}
    & \int_{1}^k f^{k, \varepsilon}_b(b) \left(b\cdot \left(\distributionlbonesidedcdfs(p)- \distributionlbonesidedcdfs(0)\right) + \int_{p}^{b\cdot \frac{e-1}{e}} s\cdot \distributionlbonesidedpdfs(s)  \d{s}\right) \d{b} = \\
    & \int_{1}^{\left(p-\varepsilon\right)\cdot \frac{e}{e-1}}\left(1-\frac{b}{e(b+\varepsilon)}\right)\frac{b}{(b+\varepsilon)^2} \d{b}
     +  \int_{(p-\varepsilon)\cdot \frac{e}{e-1}}^{k} \frac{1}{(b+\varepsilon)^2}b\left(\frac{b}{e(b-p+\varepsilon)}-\frac{b}{e(b+\varepsilon)}\right) \d b \\
     + &  \int_{(p-\varepsilon)\cdot \frac{e}{e-1}}^{k} \frac{1}{(b+\varepsilon)^2}\int_{p}^{b\cdot \frac{e-1}{e}+\varepsilon}s\cdot \frac{b}{e(b+\varepsilon -s)^2}  \d{s}] \d{b}.\\
\end{align*}

Then, we bound each expression and take its limit as $\varepsilon$ approaches $0$. Again, by the dominated convergence theorem, we can swap the order of the integral and the limit operator$^{\ref{dominated-conv-thm-footnote}}$. 
\begin{align*}
    & \lim_{\varepsilon \to 0^+}\int_{1}^{(p-\varepsilon)\cdot \frac{e}{e-1}}\left(1-\frac{b}{e(b+\varepsilon)}\right)\frac{b}{(b+\varepsilon)^2} \d{b} = \left(1-\frac{1}{e}\right)\lnn{\frac{pe}{e-1}}.\\ 
    & \lim_{\varepsilon \to 0^+} \int_{(p-\varepsilon)\cdot \frac{e}{e-1}}^{k} 
    \frac{1}{(b+\varepsilon)^2}b\left(\frac{b}{e(b-p+\varepsilon)}- \frac{b}{e(b+\varepsilon)}\right) \d b \begin{aligned}[t] & = \int_{p\cdot \frac{e}{e-1}}^{k} 
    \frac{1}{b^2}b\left(\frac{b}{e(b-p)}- \frac{b}{eb}\right) \d b \\
    & = \frac{1}{e}\left(\lnn{\frac{k-p}{\frac{p}{e-1}}} -\lnn{\frac{k}{p\frac{e}{e-1}}}\right) \\
    & = \frac{1}{e}\left(\lnn{k-p} +\lnn{e-1} -\ln k + \lnn{\frac{e}{e-1}}\right). \\
    \end{aligned} \\
    & \lim_{\varepsilon \to 0^+}  \int_{(p-\varepsilon)\cdot \frac{e}{e-1}}^k  \begin{aligned}[t] & 
 \frac{1}{(b+\varepsilon)^2} \int_{p}^{b\cdot \frac{e-1}{e}+\varepsilon}s\cdot \frac{b}{e(b+\varepsilon -s)^2}  \d{s} \d b \\
    & \leq  \lim_{\varepsilon \to 0^+}  \int_{(p-\varepsilon)\cdot \frac{e}{e-1}}^k \frac{1}{(b+\varepsilon)^2} \int_{0}^{b\cdot \frac{e-1}{e}+\varepsilon}s\cdot \frac{b}{e(b+\varepsilon -s)^2}  \d{s} \d b \\
    & = \int_{p\cdot \frac{e}{e-1}}^k \frac{1}{b^2} \int_{0}^{b\cdot \frac{e-1}{e}}s\cdot \frac{b}{e(b-s)^2}  \d{s} \d b \\
    & = \frac{1}{e}\int_{p\cdot \frac{e}{e-1}}^k \frac{1}{b} [\eval{\frac{b}{b-s} + \lnn{b-s}]}_0^{b\cdot \frac{e-1}{e}}\\
    &= \frac{1}{e}\int_{p\cdot \frac{e}{e-1}}^k \frac{1}{b}\left(e-1 +\lnn{\frac{1}{e}}\right) \d b  = \frac{e-2}{e}\lnn{\frac{k}{\frac{pe}{e-1}}}.\\
    \end{aligned} 
\end{align*}

Now, by summing all the expressions together, we get that as $\varepsilon$ approaches $0$, the welfare of $x_p$ approaches:
\begin{multline*}
    1+\frac{1}{k} + \frac{1}{e}\lnn{\frac{k}{p}}+ \left(1-\frac{1}{e}\right)\lnn{\frac{ep}{e-1}} + \frac{1}{e}\left(\lnn{k-p}-\ln k + \lnn{e-1} +\lnn{\frac{e}{e-1}}\right)  \\
    + \left(1-\frac{2}{e}\right)\left(\ln k -\ln p -\lnn{\frac{e}{e-1}}\right) = 1+\frac{1}{k} + \ln k (1-\frac{2}{e}) +\frac{1}{e} \lnn{k-p}+ \ln p \left(-\frac{1}{e}+ 1-\frac{1}{e}-1+\frac{2}{e}\right) \\
     + \lnn{\frac{e}{e-1}}\left(1-\frac{1}{e}+\frac{1}{e}-1+\frac{2}{e}\right)+\frac{\lnn{e-1}}{e} \underbrace{\leq}_{p \geq 0} 1+\frac{1}{k} + (1-\frac{1}{e})\ln k + \frac{2}{e}\lnn{\frac{e}{e-1}}+\frac{\lnn{e-1}}{e} \\
    \leq 1+\frac{1}{k}+ 0.6 + (1-\frac{1}{e})\ln k.\\
\end{multline*}

Thus, as $\varepsilon$ approaches $0$, the approximation ratio approaches $\frac{\ln k + 1+\frac{1}{k}}{(1-\frac{1}{e})\ln k +0.6+ 1+\frac{1}{k}}$. As $k$ approaches $\infty$, the upper bound approaches $\frac{e}{e-1}$.

\section{Impossibilities for Bayesian Mechanisms}\label{sec-impossibilities-bayesian}

In this section, we prove several impossibility results for Bayesian incentive-compatible mechanisms. Specifically, in Subsection~\ref{2-2-ind}, we prove a lower bound of $1.113$ for the welfare of independent distributions; in Subsection~\ref{k-k-cor-lb} we prove a lower bound of $1+\frac{\ln 2}{2}$ for the welfare of correlated distributions, and in Subsection~\ref{gft-section} we prove a lower bound of $2$ for the gains from trade of independent distributions.

For the welfare lower bounds (Subsection~\ref{2-2-ind}, and Subsection~\ref{k-k-cor-lb}), the approximation ratio is better than the ratio of $ 1.071$ that is obtained by Blumrosen and Mizrahi \cite{BM}.
In the gains from trade lower bound (Subsection~\ref{gft-section}), the approximation ratio is better than the ratio of $\frac{e}{2} \approx 1.359$, also in \cite{BM}. However, our results only apply to deterministic Bayesian mechanisms, whereas their results apply to randomized ones. We note that all major mechanisms considered in the literature (e.g., fixed price auction, buyer and seller offering mechanisms) are deterministic\footnote{Our impossibilities also apply to mechanisms which are a probability distribution over deterministic ex-post individually rational mechanisms, like the random offerer mechanism \cite{GDTCA}. To see this, consider such a mechanism $A$. A simple averaging argument shows that in the support of $A$, there must be a (deterministic and ex-post individually rational) mechanism $A'$ that its approximation ratio is at least the approximation ratio of $A$. Our impossibilities directly apply to $A'$, and the approximation ratio of $A$ is no better than the approximation ratio of $A'$.}.

In the proofs, we rely on representing joint distributions and allocation rules in tables notation. Thus, this section starts with Subsection~\ref{ill-rules}, which explains this notation. The proofs of Section~\ref{k-k-cor-lb} and Section~\ref{gft-section} share a similar structure. Both rely on $L$-shaped distributions, presented and analyzed in Section~\ref{L-shape-sec}.

\subsection{Preparations: Representing Distributions and Mechanisms}\label{ill-rules}

We consider only discrete distributions in this section (since this section aims to prove impossibilities, this only strengthens our results). We represent a discrete (joint) distribution using tables: each row corresponds to one of the possible values of the buyer, and each column corresponds to one of the possible values of the seller. Each cell $(s,b)$ corresponds to an instance in the support of the distribution. The value in the cell $(s,b)$ is the probability the instance $(s,b)$ is realized. We represent an allocation rule for the distribution by adding $*$ to a cell $(s,b)$ if the item is traded in the instance $(s,b)$. 

\begin{example}\label{2*2-ex}
An illustration of a distribution where both the seller and the buyer have support of size 2.

\begin{figure}[H]
    \centering
\renewcommand{\arraystretch}{1.4}
\begin{NiceTabularX}{6cm}{{@{}YYY@{}}}[hvlines]
\diagbox{ buyer}{seller} &$s_1$  & $s_2$  \\
$b_1$ &  $\prbuyervalue{1}^*$ & $\prbuyervalue{2}$\\
$b_2$ & $\prbuyervalue{3}^*$ & $\prbuyervalue{4}^*$ \\
\end{NiceTabularX}
\caption{An illustration of a distribution and an allocation function. Here, the buyer has two possible values $b_1 < b_2$, and the seller also has two possible values $s_1<s_2$. The probability of the instance $(s_1, b_1)$ is $\prbuyervalue{1}$, the probability of the instance $(s_2, b_1)$ is $\prbuyervalue{2}$, the probability of the instance $( s_1, b_2)$ is $\prbuyervalue{3}$, and the probability of the instance $(s_2, b_2)$ is $\prbuyervalue{4}$ ($\prbuyervalue{1}+\prbuyervalue{2}+\prbuyervalue{2}+\prbuyervalue{4} =1$). The item is traded in all instances except $(s_2, b_1)$, as is indicated by the symbol $*$ added to a cell if the item is traded in the instance that corresponds to that cell.
}
\label{2*2-fig}
\end{figure}
\end{example}

\subsection{Warm-Up: Impossibilities for Welfare Approximation via $2\times 2$ 
Distributions}\label{2-2-ind}



We now prove an impossibility result of Bayesian incentive-compatible mechanisms for independent distribution, achieved by distributions where each player has only two possible values.



Consider the family of distributions depicted in Figure~\ref{ind:all-allocations}. In every distribution that belongs to the family, the values of the buyer and the seller are independent. The buyer's value is $1$ with probability $\prbuyervalue{1}$ and is $b_2$ with probability $\prbuyervalue{2} = 1- \prbuyervalue{1}$. The seller's value is $0$ with probability $q_1$ and is $s_2$ with probability $q_2 = 1- q_1$. We have that $b_2>s_2>1>0$. The welfare maximizing allocation function of such distributions is depicted in Figure~\ref{ind:opt-allocation}. Every such distribution is called \emph{simple}.

\begin{theorem}\label{ind:2-2-lb-thm}
    There is a simple distribution $\mathcal F$ such that no Bayesian incentive compatible mechanism provides an approximation ratio better than $1.113$ for the distribution $\mathcal{F}$.
\end{theorem}

Towards proving this theorem, fix a simple distribution $\mathcal F$. We now consider the possible allocation functions for $\mathcal{F}$.
Since $\mathcal{F}$ has $4$ instances in its support, there are $2^4 = 16$ possible allocation functions. However, in one of these instances $(s_2, b_1)$, trade cannot occur since $b_1 < s_2$, resulting in $2^3=8$ possible allocation functions.

\begin{figure}[H]
     \centering
     \begin{subfigure}[]{0.4\textwidth}
         \centering
         \renewcommand{\arraystretch}{1.4}
         \begin{NiceTabularX}{6cm}{YYY}[hvlines]
    \diagbox{ buyer}{seller} &$0$  & $s_2$  \\
    $1$ &  ${\prbuyervalue{1}\cdot q_1}^*$ & ${\prbuyervalue{1}\cdot q_2}$\\
    $b_2$ & ${\prbuyervalue{2}\cdot q_1}^*$ & ${\prbuyervalue{2}\cdot q_2^*}$ \\
    \end{NiceTabularX}
         \caption{}
         \label{ind:opt-allocation}
     \end{subfigure}
     \hspace{1em}
     \vspace{1em}
     \begin{subfigure}[]{0.4\textwidth}
         \centering
         \renewcommand{\arraystretch}{1.4}
         \begin{NiceTabularX}{6cm}{YYY}[hvlines]
    \diagbox{ buyer}{seller} &$0$  & $s_2$  \\
    $1$ &  ${\prbuyervalue{1}\cdot q_1}^*$ & ${\prbuyervalue{1}\cdot q_2}$\\
    $b_2$ & ${\prbuyervalue{2}\cdot q_1}$ & ${\prbuyervalue{2}\cdot q_2}^*$ \\
    \end{NiceTabularX}
         \caption{}
         \label{ind:diag-allocation}
     \end{subfigure}
     \begin{subfigure}[]{0.4\textwidth}
         \centering
         \renewcommand{\arraystretch}{1.4}
         \begin{NiceTabularX}{6cm}{YYY}[hvlines]
    \diagbox{ buyer}{seller} &$0$  & $s_2$  \\
    $1$ &  ${\prbuyervalue{1}\cdot q_1}^*$ & ${\prbuyervalue{1}\cdot q_2}$\\
    $b_2$ & ${\prbuyervalue{2}\cdot q_1}^*$ & ${\prbuyervalue{2}\cdot q_2}$ \\
    \end{NiceTabularX}
         \caption{}
         \label{ind:col-allocation}
     \end{subfigure}
     \hspace{1em}
     \begin{subfigure}[]{0.4\textwidth}
         \centering
         \renewcommand{\arraystretch}{1.4}
        \begin{NiceTabularX}{6cm}{YYY}[hvlines]
    \diagbox{ buyer}{seller} &$0$  & $s_2$  \\
    $1$ &  ${\prbuyervalue{1}\cdot q_1}$ & ${\prbuyervalue{1}\cdot q_2}$\\
    $b_2$ & ${\prbuyervalue{2}\cdot q_1}^*$ & ${\prbuyervalue{2}\cdot q_2}^*$ \\
    \end{NiceTabularX}
         \caption{}
         \label{ind:row-allocation}
     \end{subfigure}     
        \caption{All possible allocation functions of $\mathcal{F}$ that trade the item in at least two instances.}
        \label{ind:all-allocations}
\end{figure}

Observe that any of the four allocation functions that trade the item in at most one instance will result in lower welfare compared to at least one of the allocation functions shown in Figures \ref{ind:col-allocation} and \ref{ind:row-allocation}. Specifically, the allocation functions in figures \ref{ind:1-alloc-1} and \ref{ind:1-alloc-2} have lower welfare than the allocation in Figure~\ref{ind:col-allocation}, and the allocation functions in figures \ref{ind:1-alloc-3} and \ref{ind:1-alloc-0} have lower welfare than the allocation in Figure~\ref{ind:row-allocation}. Therefore, it is sufficient to obtain a lower bound on the approximation ratio of the allocation functions shown in Figures \ref{ind:col-allocation} and \ref{ind:row-allocation} to obtain a lower bound on the approximation of all the allocation functions shown in Figure~\ref{ind:all-allocations-at-most-1}.

Hence, from now on, we consider only the four allocation functions depicted in Figure~\ref{ind:all-allocations}, as the welfare of any other implementable allocation function for $\mathcal{F}$ is at most the welfare of one of those four functions. 


We choose the values of $b_2, s_2, \prbuyervalue{1}, \prbuyervalue{2}, q_1, q_2$ so that the welfare maximizing allocation function, depicted in Figure~\ref{ind:opt-allocation}, is not implementable by a Bayesian incentive compatible mechanism (Lemma~\ref{ind-no-opt-cond}). Our choice of values also rules out the existence of a Bayesian implementation of the allocation rule in Figure~\ref{ind:diag-allocation} (Lemma~\ref{ind-no-diag-allocation-cond}). Note that the allocation rules in which the item is traded only when $b=b_2$ (depicted in Figure~\ref{ind:row-allocation}) or only when $s=0$ (depicted in Figure~\ref{ind:col-allocation}) are always implementable. In fact, they are implementable by a fixed price mechanism: the former by fixed price of $b_2$ (or $s_2$), and the latter by fixed price of $0$ (or $1$).

Hence, we are left with the allocation functions depicted in Figure~\ref{ind:col-allocation}, Figure~\ref{ind:row-allocation}. We choose the values of the parameters such that the welfare of these two allocation functions will be equal (Lemma~\ref{ind-same-welfare}) and as low as possible. Together, we will get that for the distribution $\mathcal{F}$, there is no Bayesian incentive compatible mechanism that provides an approximation ratio better than $1.113$ (Subsection \ref{subsec-ind:2-2-lb-thm}). 

\begin{figure}[H]
     \centering
     \begin{subfigure}[]{0.4\textwidth}
         \centering
         \renewcommand{\arraystretch}{1.4}
         \begin{NiceTabularX}{6cm}{YYY}[hvlines]
    \diagbox{ buyer}{seller} &$0$  & $s_2$  \\
    $1$ &  ${\prbuyervalue{1}\cdot q_1}^*$ & ${\prbuyervalue{1}\cdot q_2}$\\
    $b_2$ & ${\prbuyervalue{2}\cdot q_1}$ & ${\prbuyervalue{2}\cdot q_2}$ \\
    \end{NiceTabularX}
         \caption{}
         \label{ind:1-alloc-1}
     \end{subfigure}
     \hspace{1em}
     \vspace{1em}
     \begin{subfigure}[]{0.4\textwidth}
         \centering
         \renewcommand{\arraystretch}{1.4}
         \begin{NiceTabularX}{6cm}{YYY}[hvlines]
    \diagbox{ buyer}{seller} &$0$  & $s_2$  \\
    $1$ &  ${\prbuyervalue{1}\cdot q_1}$ & ${\prbuyervalue{1}\cdot q_2}$\\
    $b_2$ & ${\prbuyervalue{2}\cdot q_1}^*$ & ${\prbuyervalue{2}\cdot q_2}$ \\
    \end{NiceTabularX}
         \caption{}
         \label{ind:1-alloc-2}
     \end{subfigure}
     \begin{subfigure}[]{0.4\textwidth}
         \centering
         \renewcommand{\arraystretch}{1.4}
         \begin{NiceTabularX}{6cm}{YYY}[hvlines]
    \diagbox{ buyer}{seller} &$0$  & $s_2$  \\
    $1$ &  ${\prbuyervalue{1}\cdot q_1}$ & ${\prbuyervalue{1}\cdot q_2}$\\
    $b_2$ & ${\prbuyervalue{2}\cdot q_1}$ & ${\prbuyervalue{2}\cdot q_2}^*$ \\
    \end{NiceTabularX}
         \caption{}
         \label{ind:1-alloc-3}
     \end{subfigure}
     \hspace{1em}
     \begin{subfigure}[]{0.4\textwidth}
         \centering
         \renewcommand{\arraystretch}{1.4}
        \begin{NiceTabularX}{6cm}{YYY}[hvlines]
    \diagbox{ buyer}{seller} &$0$  & $s_2$  \\
    $1$ &  ${\prbuyervalue{1}\cdot q_1}$ & ${\prbuyervalue{1}\cdot q_2}$\\
    $b_2$ & ${\prbuyervalue{2}\cdot q_1}$ & ${\prbuyervalue{2}\cdot q_2}$ \\
    \end{NiceTabularX}
         \caption{}
         \label{ind:1-alloc-0}
     \end{subfigure}     
        \caption{All allocation functions of $\mathcal{F}$ that trade the item in at most one instance.}
        \label{ind:all-allocations-at-most-1}
\end{figure}

\begin{lemma} \label{ind-no-opt-cond}
    If $s_2-\frac{\prbuyervalue{1}}{\prbuyervalue{2}}> \frac{q_2}{q_1}(b_2-s_2) +1$, then the welfare-maximizing allocation function (depicted in Figure~\ref{ind:opt-allocation}) is not implementable by a Bayesian incentive compatible mechanism. 
\end{lemma}

\begin{proof}[Proof of Lemma~\ref{ind-no-opt-cond}]
Let $M=(x,p)$ be a mechanism for $\mathcal F$ that implements the welfare-maximizing allocation function. Let $p = p(0, b_2)$. The expected profit of the seller when his value is $0$ and he follows his equilibrium strategy is $\prbuyervalue{1} \cdot p(0, 1) + \prbuyervalue{2}\cdot p $. The expected profit of the seller when his value is $0$ and he plays the equilibrium strategy of $s_2$ is $ \prbuyervalue{2}\cdot p(s_2, b_2)$. The mechanism is incentive compatible for the seller, and thus $\prbuyervalue{1} \cdot p(0,1) + \prbuyervalue{2}\cdot p \geq \prbuyervalue{2}\cdot p(s_2, b_2)$ (Inequality (\ref{seller-ic}), for $s=0$ and $s'=s_2$,).
By individual rationality, the price is at most the value of the buyer $p(0,1)\leq 1$. Also, the price is at least the value of the seller $p(s_2, b_2) \geq s_2$. Together, we get:  
\begin{equation}\label{ind-opt-lemma-p-first-constraint}
\begin{split}
     \prbuyervalue{1} + \prbuyervalue{2}\cdot p  \geq \prbuyervalue{1} \cdot p(0,1) + & \prbuyervalue{2}\cdot p  \geq \prbuyervalue{2}\cdot p(s_2, b_2) \geq \prbuyervalue{2}\cdot s_2 \\
    p  & \geq s_2 -\frac{\prbuyervalue{1}}{\prbuyervalue{2}}.
\end{split}
\end{equation}
Similarly, by Bayesian incentive compatibility for the buyer (Inequality (\ref{buyer-ic}) with $b=b_2$ and $b'=1$):

\begin{equation}\label{ind-opt-lemma-p-second-constraint}
\begin{split}
   b_2 -q_1 \cdot p - q_2\cdot s_2 \underbrace{\geq}_{\substack{\text{individual} \\ \text{rationality}}} b_2 -q_1 \cdot p - q_2\cdot p(s_2,b_2) & \geq b_2\cdot q_1 -q_1\cdot p(0, 1) \underbrace{\geq}_{\substack{\text{individual} \\ \text{rationality}}} q_1\cdot b_2 -q_1\\
    \frac{q_2}{q_1}(b_2-s_2) + 1  \geq p.
\end{split}
\end{equation}

Inequalities~\eqref{ind-opt-lemma-p-first-constraint} and ~\eqref{ind-opt-lemma-p-second-constraint} imply that $ s_2 -\frac{\prbuyervalue{1}}{\prbuyervalue{2}}  \leq \frac{q_2}{q_1}(b_2-s_2) + 1$. 




\end{proof}

\begin{lemma} \label{ind-no-diag-allocation-cond}
    If $s_2-\frac{\prbuyervalue{1}}{\prbuyervalue{2}}> \frac{q_2}{q_1}(b_2-s_2) +1 $ then the allocation rule depicted in Figure~\ref{ind:diag-allocation} is not implementable by a Bayesian incentive compatible mechanism. 
\end{lemma}

\begin{proof}[Proof of Lemma~\ref{ind-no-opt-cond}]
    Let $M=(x,p)$ be a Bayesian incentive compatible mechanism that implements the allocation rule depicted in Figure~\ref{ind:diag-allocation}. Inequality~(\ref{buyer-ic}) with $b=b_2, b'=1$ implies that:

\begin{equation*}\label{ind-diag-lemma-p-second-constraint}
\begin{split}
   b_2\cdot q_2 - q_2\cdot s_2 \underbrace{\geq}_{\substack{\text{individual} \\ \text{rationality}}} b_2\cdot q_2 - q_2\cdot p(s_2,b_2)  \geq & b_2\cdot q_1 -q_1\cdot p(0, 1) \underbrace{\geq}_{\substack{\text{individual} \\ \text{rationality}}} q_1\cdot b_2 -q_1\\
    \frac{q_2}{q_1}(b_2-s_2)  \geq & b_2 -1.
 \end{split}
\end{equation*}

Together with the assumption $s_2-\frac{\prbuyervalue{1}}{\prbuyervalue{2}}> \frac{q_2}{q_1}(b_2-s_2) +1 $
we get that $s_2 -\frac{\prbuyervalue{1}}{\prbuyervalue{2}} \geq b_2$, and this is a contradiction as $b_2> s_2$, and thus the allocation rule  depicted in Figure~\ref{ind:diag-allocation} is not implementable by a Bayesian incentive compatible mechanism. 
\end{proof}

\begin{lemma} \label{ind-same-welfare}
    For $\prbuyervalue{1}\cdot q_1=\prbuyervalue{2}\cdot q_2(b_2-s_2)$, the expected welfare of the allocation rules of Figure~\ref{ind:col-allocation} and in Figure~\ref{ind:row-allocation} are equal.
\end{lemma}

\begin{proof}[Proof of Lemma~\ref{ind-same-welfare}] 
    The expected welfare of the allocation rule in Figure~\ref{ind:col-allocation} is:
    
\begin{equation}\label{ind:col-welfare}
    \prbuyervalue{1}\cdot q_1 + \prbuyervalue{1}\cdot q_2 \cdot s_2 + \prbuyervalue{2}\cdot q_1 \cdot b_2 + \prbuyervalue{2}\cdot q_2 \cdot s_2. \nonumber
\end{equation}
The expected welfare of the allocation rule in Figure~\ref{ind:row-welfare} is:
\begin{equation}\label{ind:row-welfare}
    \prbuyervalue{1}\cdot q_2 \cdot s_2 + \prbuyervalue{2}\cdot q_1 \cdot b_2 + \prbuyervalue{2}\cdot q_2 \cdot b_2.\nonumber
\end{equation}           
The two equations are equal when $\prbuyervalue{1}\cdot q_1=\prbuyervalue{2}\cdot q_2(b_2-s_2)$. 
\end{proof}


\subsubsection{Concluding the Proof of Theorem~\ref{ind:2-2-lb-thm}}\label{subsec-ind:2-2-lb-thm}
    We choose values for $b_2, s_2, \prbuyervalue{1}, \prbuyervalue{2}, q_1, q_2$ so that the conditions of Lemmas~\ref{ind-no-opt-cond}, \ref{ind-no-diag-allocation-cond}, and \ref{ind-same-welfare} are met.
    If the conditions hold, by these lemmas, the approximation ratio of every Bayesian incentive compatible mechanism for $\mathcal{F}$ is at least the approximation ratio of the mechanism depicted in Figure~\ref{ind:row-allocation}. We, therefore, choose the values for the parameters so that this approximation ratio is maximized, i.e., this expression is maximized:
     \begin{equation*}
        \frac{\prbuyervalue{1}\cdot q_1+\prbuyervalue{1}\cdot q_2\cdot s_2 + \prbuyervalue{2}b_2}{\prbuyervalue{1}\cdot q_2\cdot s_2 + \prbuyervalue{2}b_2} = 1+ \frac{\prbuyervalue{1}\cdot q_1}{\prbuyervalue{1}\cdot q_2\cdot s_2 + \prbuyervalue{2}b_2}.
    \end{equation*}
     Rearranging the expressions in the lemmas, we get:
    \begin{align}\label{ind-restrictions}
        s_2 & > b_2\cdot q_2 + \frac{q_1}{\prbuyervalue{2}} \text{ and } b_2 = s_2 +\frac{\prbuyervalue{1}\cdot q_1}{\prbuyervalue{2}\cdot q_2}\nonumber \\
        s_2 &> \frac{\prbuyervalue{1}+1}{\prbuyervalue{2}} \text{ and } b_2 = s_2 +\frac{\prbuyervalue{1}\cdot q_1}{\prbuyervalue{2}\cdot q_2}. \nonumber
    \end{align}
    Set 
    $s_2 =  b_2\cdot q_2 + \frac{q_1}{\prbuyervalue{2}} + \varepsilon$, $b_2 = \frac{1}{\prbuyervalue{2}}+\frac{\prbuyervalue{1}}{\prbuyervalue{2}\cdot q_2}+\frac{\varepsilon}{q_1}$, $\prbuyervalue{1} = 0.57, q_1 = 0.716$.\footnote{The values of $\prbuyervalue{1}, q_1$ are chosen to maximize the function $f(\prbuyervalue{1},q_1)= 1+ \frac{\prbuyervalue{1}\cdot q_1 }{\frac{\prbuyervalue{1}(1-q_1)(1+\prbuyervalue{1})}{(1-\prbuyervalue{1})}+ \frac{\prbuyervalue{1}\cdot q_1}{(1-q_1)}+\prbuyervalue{1}+1}$.}.
    As $\varepsilon$ approaches $0$, $b_2$ approaches 
    $\frac{1}{\prbuyervalue{2}}+\frac{\prbuyervalue{1}}{\prbuyervalue{2}\cdot q_2} \approx 6.993$, $s_2$ approaches$ \frac{1+\prbuyervalue{1}}{\prbuyervalue{2}} \approx 3.651$ and the approximation ratio $1+ \frac{\prbuyervalue{1}\cdot q_1}{\prbuyervalue{1}\cdot q_2\cdot s_2 + \prbuyervalue{2}b_2}$ approaches 1.113.

\subsection{$L$-Shaped Distributions} \label{L-shape-sec}

We now define and analyze $L$-shaped distributions, which are used in Section~\ref{k-k-cor-lb} and Section~\ref{gft-section}. The heart of the proofs in these sections is Claim~\ref{standard-form-claim}. In this section, we define and analyze the basic properties of $L$-shaped distributions and their standard form (Section~\ref{standard-L-shape-sec}) and prove Claim  ~\ref{standard-form-claim}.

Consider the family of $L$-shaped distributions illustrated in Figure~\ref{L-shaped-k-k}. In each distribution in this family, the size of the support of the buyer's and seller's values is $k$. The buyer's support is $b_1 = 1 < b_2 < \dots < b_k$, and the seller's support is $s_1=0 < s_2< \dots < s_k$.
This family of distribution is called $L$-shaped as every distribution in this family satisfies a condition that restricts the trade only to instances inside the $L$-shape.
Consider the set of instances that are outside the $L$-shape, i.e.,  ${\overline{L}} = \{ (s_j, b_i) | s_j \neq 0 \text{ and } b_i \neq b_k\}$. Then, every distribution in the family either has $0$ probability for every instance in $\overline{L}$, or in every instance in $\overline{L}$, the seller's value is larger than the buyer's value. Either way, a trade cannot occur in the instances in $\overline{L}$.

We consider $L$-shaped distributions as they are easier to analyze. Every allocation rule for an $L$-shaped distribution is associated with two threshold values: $b_i \in \{ b_1, \dots b_k \} \cup \{\infty\}$ and $s_j \in \{s_2, \dots, s_k\} \cup \{0\}$ (Observation~\ref{k-k-threshold-obs}). The allocation in the instances $\{ (0,1), (0,b_2), \dots, (0, b_{k-1})\}$ is determined by the threshold $b_i$: if the buyer's value is below $b_i$, the item is not traded; if the buyer's value is at least $b_i$, the item is traded. Similarly, the allocation in the instances $\{ (s_2, b_k), \dots, (s_k, b_k)\}$ is determined by the threshold $s_j$: if the seller's value is above $s_j$, the item is not traded; if the seller's value is at most $s_j$, the item is traded. This property reduces the number of potential allocation rules of Bayesian incentive-compatible mechanisms. Furthermore, it provides a simple characterization of such allocation rules.

\begin{figure}[H]
    \centering
    \renewcommand{\arraystretch}{1.4}
\begin{NiceTabularX}{11cm}{YYYYYY}[hvlines]
\CodeBefore
  \begin{tikzpicture}
    \cellcolor{tradepossiblecolor}{2-2,3-2, 4-2, 5-2, 6-2, 6-3, 6-4, 6-5, 6-6}
  \end{tikzpicture}
\Body
    \diagbox{ buyer}{seller } &$s_1 = 0$  & $s_2$  & $\dots$ & $s_{k-1}$ & $s_k$  \\
    $b_1 = 1$ &  ${+}^*$ & $?$ & $\dots$ & $?$ & $?$\\
    $b_2$ & ${+}^*$ & $?$ & $\dots$ & $?$ & $?$\\
    $\vdots$ &  ${\vdots}$ & $\vdots$ & $\dots$ & $\vdots$ & $\vdots$\\
    $b_{k-1}$ &   ${+}^*$ & $?$ & $?$ & $?$ & $?$\\
    $b_k$ &   ${+}^*$ & ${+}^*$ & $\dots$ & ${+}^*$ & ${+}^*$\\
    \end{NiceTabularX}
    \caption{An L-shaped distribution and the allocation function that maximizes its welfare (and its gains from trade).
    Each cell corresponding to an instance with positive probability is marked by $+$. A cell that corresponds to an instance that might have $0$ probability is marked by $?$. 
    The cells where trade occurs in the optimal allocation rule are colored in \tradepossiblecolorname.}
    \label{L-shaped-k-k}
\end{figure}




\subsubsection{Characterizing Implementable Allocation Functions for L-shaped Distributions}

We now characterize the allocation rules that are implementable by Bayesian incentive-compatible mechanisms in L-shaped distributions.

\begin{observation}\label{k-k-threshold-obs}
    Suppose $M=(x,p)$ is a Bayesian incentive compatible mechanism for an L-shaped distribution $\mathcal{F}_k$. Then, there exist two \emph{threshold values} $b_i \in \{ b_1, \dots b_k\} \cup \{\infty\}$ and $s_j \in \{{s_2, \dots, s_k\}} \cup \{0\}$ that determine the allocation in the instances $\{(0,1), (0,b_2), \dots, (0, b_{k-1})\}$ and $\{(s_2, b_k), \dots, (s_k, b_k)\}$. 
    
    The allocation in the instances $\{(0,1), (0,b_2), \dots, (0, b_{k-1})\}$ is determined by the threshold $b_i$: if the buyer's value is below $b_i$, the item is not traded; if the buyer's value is at least $b_i$, the item is traded. Similarly, the allocation in the instances $\{(s_2, b_k), \dots, (s_k, b_k)\}$ is determined by the threshold $s_j$: if the seller's value is above $s_j$, the item is not traded; if the seller's value is at most $s_j$, the item is traded. 
    
    Moreover, the payment in the instances $\{(0,1), (0,b_2), \dots, (0, b_{k-1})\}$, where the item is traded, is the same. Similarly, the payment in the instances $\{(s_2, b_k), \dots, (s_k, b_k)\}$, where the item is traded, is the same.
\end{observation}

By Observation~\ref{k-k-threshold-obs}, the only instance inside the $L$-shape whose allocation is not determined by the thresholds $b_t, s_t$ is $(0, b_k)$. Hence, each implementable allocation rule for an $L$-shaped distribution can be characterized by three parameters: $b_t$, $s_t$, and $x( 0,b_k)$. 

Note that while every allocation rule that can be implemented by a Bayesian incentive-compatible mechanism can be described using these three parameters, not every allocation rule described by these parameters can be implemented by a Bayesian incentive compatible mechanism.

\begin{proof}[Proof of Observation~\ref{k-k-threshold-obs}]
    We prove the claims regarding the threshold for the buyer. The proof for the threshold of the seller is very similar.
     Consider an L-shaped distribution $\mathcal{F}_k$. Whenever the buyer's value is $b \in \{ b_1, \dots b_{k-1} \}$, a trade can only occur when the seller's value is $0$. 
     Let $b_i \in\{ b_1, \dots b_{k-1} \} $ be the smallest value in $\{ b_1, \dots b_{k-1} \}$ such that the item is traded in the instance $(0, b_i)$ by $M$. If the item is not traded in the instances $\{(0,1), (0,b_2), \dots, (0, b_{k-1})\}$ we say that the threshold is $\infty$.
     
     Assume that $b_i$ is not a threshold, i.e., there exists $b_j \in \{b_{i+1}, \dots, b_{k-1}\}$ such that in the instance $(0, b_j)$, the item is not traded by $M$. The expected profit of the buyer when his value is $b_j$ and he follows his equilibrium strategy is $\Pr(s=0|\, b=b_j)\cdot (x(0, b_j)\cdot b_j - p(0, b_j)) = -\Pr(s=0|\, b=b_j)p(0, b_j)$, which equals $0$ as by individual rationality if the buyer does not get the item he pays nothing. 
     The expected profit of the buyer when his value is $b_j$ and he plays the equilibrium strategy of $b_i$ is $\Pr(s=0|\, b=b_j)(x(0, b_i)\cdot b_j - p(0, b_i)) = \Pr(s=0|\, b=b_j)(b_j -p(0, b_i))$, which is grater than $0$ as by individual rationality $p(0, b_i) \leq b_i$ and since $ b_j> b_i$. Note that incentive constraints are violated since $\Pr(s=0|\, b=b_j)(x(0, b_j)\cdot b_j - p(0, b_j)) = -\Pr(s=0|\, b=b_j)\cdot p(0, b_j)  < \Pr(s=0|\, b=b_j)(x(0, b_i)\cdot b_j - p(0, b_i)) = \Pr(s=0|\, b=b_j)(b_j -p(0, b_i))$. Thus, the item is traded in the instance $(0, b_j)$ by $M$. 
     Furthermore, observe that it is not enough that the item is traded in $(0, b_j)$; it should also be traded for the same price. Otherwise, for two values $b,b'$ that are at least $b_i$ with $p(0, b) < p(0,b')$, a buyer with value $b'$ will prefer the equilibrium strategy of $b$.
\end{proof}


\begin{figure}[H]
    \centering
    \renewcommand{\arraystretch}{1.6}
\begin{NiceTabularX}{12cm}{YYYYYYYY}[hvlines]
\CodeBefore
  \begin{tikzpicture}
    \cellcolor{white}{2-2,3-2, 4-2, 5-2, 6-2, 7-2, 8-2, 8-3, 8-4, 8-5, 8-6, 8-7, 8-8}
    \end{tikzpicture}
    \Body
    \diagbox{ buyer}{seller } &$s_1 = 0$  & $s_2$  & $\dots$ & $s_j$ & $s_{j+1}$ & $\dots$ & $s_k$  \\
    $b_1 = 1$ &  ${+}$ & $?$& $\dots$& $\dots$ & $\dots$ & $\dots$ & $?$\\
    $\vdots$ &  ${\vdots}$ & $\vdots$ & $\dots$ &$\dots$& $\dots$ & $\dots$ & $\vdots$\\
    $b_{i-1}$ & ${+}$ & $?$ & $\dots$ $\dots$& $\dots$ & $\dots$ & $\dots$ & $?$\\
    $b_{i}$ & ${+}^*$ & $?$ & $\dots$ $\dots$& $\dots$ & $\dots$ & $\dots$ & $?$\\
    $\vdots$ &  ${\vdots}$ & $\vdots$ & $\dots$ &$\dots$& $\dots$ & $\dots$ & $\vdots$\\
    $b_{k-1}$ &   ${+}^*$  & $?$ & $\dots$ $\dots$& $\dots$ & $\dots$ & $\dots$ & $?$\\
    $b_k$ &   ${+}$ & ${+}^*$ & $\dots$ & $+^*$& ${+}$ &$\dots$ & ${+}$\\
    \CodeAfter
    \tikz \node [rounded corners, inner ysep = -0.9mm, draw=\thsbuyercolor, very thick, fit = (2-2) (7-2)] { } ;
    \tikz \node [rounded corners, inner ysep = -0.9mm, draw=\thssellercolor, very thick, fit = (8-3) (8-8)] { } ;
    \end{NiceTabularX}

    \caption{An L-shaped distribution $\mathcal{F}_k$ and the allocation rule with the buyer's threshold $b_i$, the seller's threshold $s_j$, and no allocation in the instance $(0, b_k)$.
     $b_i$ is the threshold of the buyer for the instances in the \thsbuyercolorname  rectangle: $(0,b_1), \dots ,(0,b_{k-1})$. 
    The item is only traded in the \thsbuyercolorname rectangle when $b\geq b_i$. Furthermore, by Observation~\ref{k-k-threshold-obs}, the price is the same for every $ b_k> b\geq b_i$.
    $s_j$ is the threshold of the seller for the instances in the \thssellercolorname rectangle: $(s_2,b_k), \dots ,(s_k,b_k)$. 
    The item is only traded in the \thssellercolorname rectangle, when $s\leq s_j$. Furthermore, by Observation~\ref{k-k-threshold-obs}, the price is the same for every $ s_1 < s\leq s_j$. Finally, as the item is not traded in the instance $(0,b_k)$, there is no * in the cell corresponding to this instance.
    Each cell corresponding to an instance with positive probability is marked by $+$. A cell that corresponds to an instance that might have $0$ probability is marked by $?$.}
    \label{cor:k-k-characterization-mech}
\end{figure}

\subsubsection{Standard L-Shaped Distributions}\label{standard-L-shape-sec}

To prove Theorem~\ref{k-k-cor-welfare-thm}, we consider specific type of ``standard'' $L$-shaped distributions (Definition~\ref{standard-L-distribution}). Claim~\ref{standard-form-claim} proves a useful property of these distributions.

\begin{definition}\label{standard-L-distribution}
An $L$-shaped distribution $\mathcal{F}_k$ is in \emph{standard $L$-shaped form} if the following conditions hold:
\begin{enumerate}
    \item For every $1\leq i\leq k-1$ and every $i+1 \leq j<k$: 
    $$
    \frac{\sum_{r=2}^j\Pr(s=s_r| b=b_k)(b_k-s_j)}{\Pr(s=0|b=b_k)} +b_i < s_j -\frac{b_i\cdot \sum_{r=i}^{k-1}\Pr(b=b_r| s=0)}{\Pr(b=b_k|s=0)}.
    $$ 
    \label{infeasible-k-plus-1-allocations}
    \item For every $1\leq i\leq k-1$: 
    $
    \Pr(0, b_i)\cdot b_i = \Pr(s_{i+1}, b_k) \cdot (b_k -s_{i+1}).
    $\label{same-welfare-cond}
\end{enumerate}
\end{definition}

\begin{claim}\label{standard-form-claim}
Let $k\geq 2$. Let $\mathcal{F}_k$ be an $L$-shaped distribution in standard $L$-shaped form.
Then, the allocation rules that trade the item in exactly $k$ instances, with one of these instances being $(0, b_k)$, have the highest welfare and gains from trade among all Bayesian incentive-compatible implementable allocation rules. 
\end{claim}
The proof of Claim~\ref{standard-form-claim} uses Lemma~\ref{k-k-cor-inf-k-plus-1-alloc}, Lemma~\ref{k-k-cor-inf-k-alloc-no-p}, Lemma~\ref{k-k-cor-same-welfare-k-alloc}, and Observation~\ref{first-condition-gives-second}. We first provide their statements and the proof of Claim~\ref{standard-form-claim}. We then prove the three lemmas and the observation.

\begin{lemma}\label{k-k-cor-inf-k-plus-1-alloc}
Assume that for every $1\leq i\leq k-1$ and $i+1 \leq j<k$, it holds that: 
$$
\frac{\sum_{r=2}^j\Pr(s=s_r| b=b_k)(b_k-s_j)}{\Pr(s=0|b=b_k)} +b_i < s_j -\frac{b_i\cdot \sum_{r=i}^{k-1}\Pr(b=b_r| s=0)}{\Pr(b=b_k|s=0)}.  
$$
Then, every allocation rule $(b_i, s_j, 1)$ that trades the item in at least $k+1$ instances is not implementable by a Bayesian incentive compatible mechanism.
\end{lemma}

\begin{lemma}\label{k-k-cor-inf-k-alloc-no-p}
Assume that for every $1\leq i\leq k-1$ and $i+1 \leq j<k$, it holds that: 
$$b_i \sum_{r=i}^{k-1}\Pr(b=b_r|s=0) <\Pr(b=b_k| s=0) \cdot s_j \text{ OR } \sum_{r=2}^j\Pr(s=s_r|b=b_k)\cdot(b_k-s_j)< \Pr(s=0|b=b_k)(b_k -b_i).$$
Then, every allocation rule $(b_i, s_j, 0)$ that trades the item in at least $k$ instances is not implementable by a Bayesian incentive compatible mechanism.
\end{lemma}

\begin{lemma}\label{k-k-cor-same-welfare-k-alloc}
    Suppose that for every $1\leq i\leq k-1$, it holds that $
    \Pr(0, b_i)\cdot b_i = \Pr(s_{i+1}, b_k) \cdot (b_k -s_{i+1})$.
    Then, all allocation rules $(b_t, s_t, 1)$ that trade the item in exactly $k$ instances have the same welfare and have the same gains from trade. 
    
\end{lemma}

\begin{observation}\label{first-condition-gives-second}
    If Condition~\ref{infeasible-k-plus-1-allocations} of Definition~\ref{standard-L-distribution} holds then so does the following condition: For every $1\leq i\leq k-1$ and every $i+1 \leq j<k$:
    $$
    b_i \sum_{r=i}^{k-1}\Pr(b=b_r|s=0) <\Pr(b=b_k| s=0) \cdot s_j \text{ OR } \sum_{r=2}^j\Pr(s=s_r|b=b_k)\cdot(b_k-s_j)< \Pr(s=0|b=b_k)(b_k -b_i).
    $$
\end{observation}

\begin{proof}[Proof of Claim~\ref{standard-form-claim}]

    Consider an $L$-shaped distribution in a standard form.
    By the second condition (\ref{same-welfare-cond}) in Definition~\ref{standard-L-distribution}, and by Lemma~\ref{k-k-cor-same-welfare-k-alloc}, we have that every allocation rule $(b_t, s_t, 1)$ that trades the item in exactly $k$ instances have the same welfare and the same gains from trade, we denote this welfare by $W_k$ and this gains from trade by $GFT_k$.  
    We prove that if an allocation rule $(b_t, s_t, x( 0, b_k))$ is implementable, then its welfare is at most $W_k$, and its gains from trade is at most $GFT_k$. 
    
    Let $(b_i, s_j, 1)$ be an allocation rule that trades the item in the instance $(0,b_k)$. If it trades the item in more than $k$ instances, then, by Condition~\ref{infeasible-k-plus-1-allocations} of Definition~\ref{standard-L-distribution} and Lemma~\ref{k-k-cor-inf-k-plus-1-alloc}, it cannot be implemented by a Bayesian incentive-compatible mechanism. If it trades the item in exactly $k$ instances, its welfare is $W_k$, and its gains from trade is $GFT_k$. However, if it trades the item in less than $k$ instances, then we claim that its welfare and gains from trade are lower than those of some other allocation rule $(b'_i, s'_j, 1)$, which trades the item in $k$ instances. When $b_i = \infty$, then the allocation rule $(\infty, s_k, 1)$, trades the item in every instance that $(b_i, s_j, 1)$ does and even in instances that it does not. Similarly, when $s_j=0$, then the allocation rule $(b_1, 0, 1)$, trades the item in every instance that  $(b_i, s_j, 1)$ does and even in instances that it does not. Finally, when both $b_i \neq \infty$ and $s_j \neq 0$, then the allocation rule $(b_i, s_i, 1)$  trades the item in every instance that $(b_i, s_j, 1)$ does and even in instances that it does not, since $i > j$, as the number of instances in which the allocation rule  $(b_i, s_i, 1)$ trades the item is $k-i+j$ less than $k$.

    Let $(b_i, s_j, 0)$ be an allocation rule that does not trade the item in the instance $(0,b_k)$. If it trades the item in at least $k$ instances, then, by Condition~\ref{infeasible-k-plus-1-allocations} of Definition~\ref{standard-L-distribution}, Observation~\ref{first-condition-gives-second}, and Lemma~\ref{k-k-cor-inf-k-alloc-no-p}, it cannot be implemented by a Bayesian incentive compatible mechanism. However, if it trades the item in at most $k-1$ instances, its welfare and gains from trade are lower than the welfare and gains from trade of the allocation rule $(b_i, s_j, 1)$, which trades the item in at most $k$ instances and has a welfare of at most $W_k$ and gains from trade of at most $GFT_k$ (as discussed above).
    Therefore, the welfare and gains from trade of any implementable allocation rule $(b_t, s_t, x(0, b_k))$ is at most $W_k$ and at most $GFT_k$, respectively.
\end{proof}

\begin{proof}[Proof of Lemma~\ref{k-k-cor-inf-k-plus-1-alloc}]
    Observe that if an allocation rule $(b_i, s_j, 1)$ trades the item in at least $k+1$ instances, then $b_i \in \{ b_1, \dots, b_{k-1}\}$, and $s_j \in \{ s_2, \dots, s_k \}$. If 
    $b_i \in \{ b_1, \dots, b_{k-1}\}$ and $s_j \in \{ s_2, \dots, s_k \}$, then the allocation rule $(b_i, s_j, 1)$ trades the item in exactly $k-i+j$ instances. Finally, $k-i+j\geq k+1$ for every $1\leq i\leq k-1$ and every $i+1 \leq j<k$. Hence, we prove that for every $1\leq i\leq k-1$ and every $i+1 \leq j<k$, if:
    $$
    \frac{\sum_{r=2}^j\Pr(s=s_r| b=b_k)(b_k-s_j)}{\Pr(s=0|b=b_k)} +b_i < s_j -\frac{b_i\cdot \sum_{r=i}^{k-1}\Pr(b=b_r| s=0)}{\Pr(b=b_k|s=0)}.
    $$
    then the allocation rule $(b_i, s_j, 1)$ is not implementable by a Bayesian incentive compatible mechanism. 
    Fix $1\leq i\leq k-1$ and $ i+1 \leq j < k$ and consider the allocation rule $(b_i, s_j, 1)$ (see Figure~\ref{cor:k-k-characterization-mech}). 
    Assume that this allocation rule is implementable by a Bayesian incentive compatible mechanism $M=(x,p)$, and denote by $p$ the trade price in the instance $( 0, b_k)$. 
    By Bayesian incentive compatibility for the buyer (Inequality \ref{buyer-ic} with $b=b_k, b'=b_i$) and the seller (Inequality \ref{seller-ic} with $s=0, s'=s_j$):
    \begin{align*}
        b_k \cdot (\sum_{r=1}^j \Pr(s=s_r|b=b_k)) - \sum_{r=1}^j \Pr(s=s_r|b=b_k))\cdot p(s_r, b_k) & \geq \Pr(s=0|b=b_k)\cdot (b_k-p(0, b_i))\\
        \Pr(b=b_k | s=0)\cdot p(0, b_k) +\sum_{l=i}^{k-1} \Pr(b=b_l | s=0)\cdot p(0,b_l) & \geq \Pr(b=b_k | s=0) \cdot p(s_j, b_k).
    \end{align*}
    Recall Observation~\ref{k-k-threshold-obs}. The prices $p( 0, b_l)$ are equal and so are the prices $p(s_r, b_k)$. By individual rationality we have  $p(0, b_l) \leq b_i$ and  $p(s_j, b_k) \geq s_j$. 
    \begin{align*}
        (b_k - s_j)\sum_{r=2}^j \Pr(s=s_r| b=b_k) +\Pr(s=0|b=b_k)(b_k - p) \geq  & \Pr(s=0|b=b_k)\cdot (b_k-b_i) \\
        \Pr(b=b_k | s=0)\cdot p + \sum_{l=i}^{k-1} \Pr(b=b_l | s=0)\cdot b_i \geq &  \Pr(b=b_k | s=0) \cdot s_j     
    \end{align*}
    Using these inequalities, we can bound $p$:
        \begin{align*}
        (b_k - s_j)\cdot \sum_{r=2}^j \frac{ \Pr(s=s_r| b=b_k)}{\Pr(s=0|b=b_k)} + b_i &\geq p\\
          s_j -\sum_{l=i}^{k-1} \frac{b_i\cdot \Pr(b=b_l | s=0)}{\Pr(b=b_k|s=0)} & \leq p.
    \end{align*}

Combining the two bounds:
\begin{align*}
     (b_k - s_j)\cdot \sum_{r=2}^j \frac{ \Pr(s=s_r| b=b_k)}{\Pr(s=0|b=b_k)} + b_i \geq s_j -\sum_{l=i}^{k-1} \frac{b_i\cdot \Pr(b=b_l | s=0)}{\Pr(b=b_k|s=0)}.
\end{align*}
\end{proof}

\begin{proof}[Proof of Lemma~\ref{k-k-cor-inf-k-alloc-no-p}]
   If an allocation rule $(b_i, s_j, 0)$ trades the item in at least $k$ instances, it must be that  $b_i \in \{ b_1, \dots, b_{k-1}\}$ and $s_j \in \{ s_2, \dots, s_k \}$. If 
    $b_i \in \{ b_1, \dots, b_{k-1}\}$ and $s_j \in \{ s_2, \dots, s_k \}$, then the allocation rule $(b_i, s_j, 0)$ trades the item in exactly $k-i+j-1$ instances. Finally, $k-i+j-1\geq k$ for every $1\leq i\leq k-1$ and every $i+1 \leq j<k$. Hence, we prove that for every $1\leq i\leq k-1$ and every $i+1 \leq j<k$, if:
    $$
       b_i \cdot \sum_{r=i}^{k-1}\Pr(b=b_r|s=0) <\Pr(b=b_k| s=0) \cdot s_j \text{ OR } \sum_{r=2}^j\Pr(s=s_r|b=b_k)\cdot(b_k-s_j)< \Pr(s=0|b=b_k)(b_k -b_i).
    $$
    then the allocation rule $(b_i, s_j, 0)$ is not implementable by a Bayesian incentive compatible mechanism. 
    For every $1\leq i\leq k-1$ and $i+1\leq j<k$, consider the allocation rule $(b_i, s_j, 0)$ (depicted in Figure~\ref{cor:k-k-characterization-mech}). Assume that this allocation rule is implementable using a Bayesian incentive compatible mechanism $M=(x,p)$. Then, by Bayesian incentive compatibility for the buyer (Inequality (\ref{buyer-ic}) with $b=b_k, b'=b_i$) and the seller (Inequality \ref{seller-ic} with $s=0, s'=s_j$) we get: 
    \begin{align*}
        & \sum_{r=2}^j\Pr(s=s_r|b=b_k)\cdot(b_k-p(s_r, b_k))\geq  \Pr(s=0|b=b_k)(b_k - p( 0, b_i)) \\
        &\sum_{r=i}^{k-1}\Pr(b=b_r|s=0)\cdot p(0, b_l) \geq\Pr(b=b_k| s=0)\cdot p(s_j, b_k).
    \end{align*}
     Recall Observation~\ref{k-k-threshold-obs}. The prices $p( 0, b_l)$ are equal and so are the prices $p(s_r, b_k)$. By individual rationality we have  $p( 0, b_l) \leq b_i$ and  $p(s_j, b_k) \geq s_j$.
    \begin{align*}
         & \sum_{r=2}^j\Pr(s=s_r|b=b_k)\cdot(b_k-s_j)\geq  \Pr(s=0|b=b_k)(b_k - b_i) \\
        &\sum_{r=i}^{k-1}\Pr(b=b_r|s=0)\cdot b_i \geq \Pr(b=b_k| s=0)\cdot s_j.
    \end{align*}
    Hence, if at least one of the inequalities below is violated, there is no Bayesian incentive-compatible mechanism that implements the allocation rule $(b_i, s_j, 0)$.
\end{proof}

\begin{proof}[Proof of Lemma~\ref{k-k-cor-same-welfare-k-alloc}]
    Note that an allocation rule $(b_i, s_j, 1)$ trades the item in $k-i-j$ instances if $b_i < \infty$ and $s_j > 0$, in $j$ instances if $b_i = \infty$, and in $k-i+1$ instances if $s_j = 0$. If both $b_i = \infty$ and $s_j = 0$, the allocation rule trades the item once. Therefore, the allocation rules that sell the item exactly $k$ times are $(b_i, s_i, 1)$ for $2 \leq i \leq k-1$, $(\infty, s_k, 1)$, and $(b_1, 0, 1)$.
    To ensure that all of these allocation rules have the same welfare and same gains from trade, we require that each allocation rule $(b_i, s_i, 1)$ has the same welfare and same gains from trade as its adjacent allocation rule. Specifically, for $2 \leq i \leq k-2$, the adjacent allocation rule is $(b_{i+1}, s_{i+1}, 1)$ (see Figure~\ref{cor:k-alloc-welfare-adjacent} for an example). For $(b_{k-1}, s_{k-1}, 1)$, the adjacent allocation rule is $(\infty, s_k, 1)$. And for $(b_{1}, 0, 1)$, the adjacent allocation rule is $(b_2, s_2, 1)$.

    \begin{figure}[H]
     \centering
    \renewcommand{\arraystretch}{1.6}
    \begin{NiceTabularX}{14cm}{YYYYYYYY}[hvlines]
    \diagbox{ buyer}{seller } &$s_1 = 0$  & $s_2$  & $\dots$ & $s_i$ & $s_{i+1}$ & $\dots$ & $s_k$  \\
    $b_1 = 1$ &  ${+}$ & $?$& $\dots$& $\dots$ & $\dots$ & $\dots$ & $?$\\
    $\vdots$ &  ${\vdots}$ & $\vdots$ & $\dots$ &$\dots$& $\dots$ & $\dots$ & $\vdots$\\
    $b_{i}$ & ${+}^*$ & $?$ & $\dots$ $\dots$& $\dots$ & $\dots$ & $\dots$ & $?$\\
    $b_{i+1}$ & ${+}^{*@}$ & $?$ & $\dots$ $\dots$& $\dots$ & $\dots$ & $\dots$ & $?$\\
    $\vdots$ &  ${\vdots}$ & $\vdots$ & $\dots$ &$\dots$& $\dots$ & $\dots$ & $\vdots$\\
    $b_{k-1}$ &   ${+}^{*@}$  & $?$ & $\dots$ $\dots$& $\dots$ & $\dots$ & $\dots$ & $?$\\
    $b_k$ &   ${+}^{*@}$ & ${+}^{*@}$ & $\dots$ & ${+}^{*@}$& ${+}^@$ &$\dots$ & ${+}$\\
    \end{NiceTabularX}
    \caption{The allocation rule $(b_i, s_i, 1)$ and its adjacent allocation rule $(b_{i+1}, s_{i+1}, 1)$ as defined in the proof of Lemma~\ref{k-k-cor-same-welfare-k-alloc}. The allocations according to $(b_i, s_i, 1)$ are marked by $*$, and the allocations according to $(b_{i+1}, s_{i+1}, 1)$ are marked by $@$.
    Each cell that corresponds to an instance with positive probability is marked by $+$, and a cell that corresponds to an instance that might have $0$ probability and might not is marked by $?$.}
    \label{cor:k-alloc-welfare-adjacent}
\end{figure}
    Note that for $2 \leq i \leq k-2$, an allocation rule $(b_i, s_i, 1)$ has the same welfare and the same gains from trade as its adjacent allocation rule $(b_{i+1}, s_{i+1}, 1)$ if the following equality holds:
    \begin{align*}
        \Pr(0,b_i) \cdot b_i  = \Pr(s_{i+1}, b_k)s_{i+1}\cdot (b_k-s_{i+1}).
    \end{align*}
    Similarly, for the allocation rule $(b_{k-1}, s_{k-1}, 1)$ to have the same welfare and same gains from trade as its adjacent allocation rule $(\infty, s_k, 1)$, we need to satisfy the equation:
    \begin{align*}
        \Pr(0, b_{k-1}) \cdot b_{k-1} = \Pr(s_k, b_k)\cdot (b_k-s_k).
    \end{align*}
    For the allocation rule $(b_{1}, 0, 1)$ to have the same welfare and the same gains from trade as its adjacent allocation rule $(b_2, s_2, 1)$, we need the following equation to hold:
    \begin{align*}
        \Pr(0, b_1) \cdot b_1 = \Pr(s_2, b_k) \cdot (b_k-s_2).
    \end{align*}
    Together, for every $1\leq i\leq k-1$, we have $\Pr(0, b_i)\cdot b_i = \Pr(s_{i+1}, b_k) \cdot (b_k -s_{i+1})$, as required.
\end{proof}

\begin{proof}[Proof of Observation~\ref{first-condition-gives-second}]
Assume that Condition~\ref{infeasible-k-plus-1-allocations} of Definition~\ref{standard-L-distribution} holds. For every $1 \leq i \leq k-1$ and $i+1 \leq j < k$: 
    \begin{align*}
    \frac{\sum_{r=2}^j\Pr(s=s_r| b=b_k)(b_k-s_j)}{\Pr(s=0|b=b_k)} +b_i < & s_j -\frac{b_i\cdot \sum_{r=i}^{k-1}\Pr(b=b_r| s=0)}{\Pr(b=b_k|s=0)} \\ \iff& \\
    \sum_{r=2}^j\Pr(s=s_r| b=b_k)(b_k-s_j) < & \Pr(s=0|b=b_k) (s_j -b_i -\frac{b_i\cdot \sum_{r=i}^{k-1}\Pr(b=b_r| s=0)}{\Pr(b=b_k|s=0)}).
\end{align*}
Observe that:
\begin{align*}
     \Pr(s=0|b=b_k) (s_j -b_i -\frac{b_i\cdot \sum_{r=i}^{k-1}\Pr(b=b_r| s=0)}{\Pr(b=b_k|s=0)}) \underbrace{<}_{s_j < b_k}  \Pr(s=0|b=b_k) (b_k -b_i),    
\end{align*}
and so: $\sum_{r=2}^j\Pr(s=s_r|b=b_k)\cdot(b_k-s_j)< \Pr(s=0|b=b_k)\cdot (b_k -b_i).$
\end{proof}

\subsection{Approximating the Welfare with Correlated Values} \label{k-k-cor-lb}

In this section, we prove a lower bound for approximating the welfare.
We consider a family of correlated distributions that are $L$-shaped (Section~\ref{L-shape-sec}). 
In every distribution $\mathcal{F}_k$ (depicted in Figure~\ref{cor:k-k}) in the family, the number of instances in the support is only $2k -1$. When the buyer's value is $b \in \{ 1, b_2, \dots, b_{k-1}\}$, then the seller's value is always $0$. Similarly, when the seller's value is $s \in \{ s_2, \dots, s_k \}$, then the buyer's value is always $b_k$. The probability of the buyer's value being $b_i$ is $\prbuyervalue{i}$ for $i \in [k]$. When the buyer's value is $b_k$, the probability that the seller's value is $s_i$ is $q_i$, for $i \in [k]$. In every instance in its support, the buyer's value is larger than the seller's value (i.e., $s_k < b_k$), and so a trade improves the welfare in each instance in the support.

\begin{figure}[H]
    \centering
    \renewcommand{\arraystretch}{1.4}
\begin{NiceTabularX}{11cm}{YYYYYY}[hvlines]
\CodeBefore
  \begin{tikzpicture}
    \cellcolor{tradepossiblecolor}{2-2,3-2, 4-2, 5-2, 6-2, 6-3, 6-4, 6-5, 6-6}
  \end{tikzpicture}
\Body
    \diagbox{ buyer}{seller } &$s_1 = 0$  & $s_2$  & $\dots$ & $s_{k-1}$ & $s_k$  \\
    $b_1 = 1$ &  ${\prbuyervalue{1}}^*$ & $0$ & $\dots$ & $0$ & $0$\\
    $b_2$ & ${\prbuyervalue{2}}^*$ & $0$ & $\dots$ & $0$ & $0$\\
    $\vdots$ &  ${\vdots}$ & $\vdots$ & $\dots$ & $\vdots$ & $\vdots$\\
    $b_{k-1}$ &   ${\prbuyervalue{k-1}}^*$ & $0$ & $\dots$ & $0$ & $0$\\
    $b_k$ &   ${\prbuyervalue{k}\cdot q_1}^*$ & ${\prbuyervalue{k}\cdot q_2}^*$ & $\dots$ & ${\prbuyervalue{k}\cdot q_{k-1}}^*$ & ${\prbuyervalue{k}\cdot q_k}^*$\\
    \end{NiceTabularX}
    \caption{An L-shaped distribution $\mathcal{F}_k$ and its welfare maximizing allocation function. 
    The colored cells in \tradepossiblecolorname  correspond to instances with positive probability.}
    \label{cor:k-k}
\end{figure}

\begin{theorem} \label{k-k-cor-welfare-thm}
    Let $k\geq 2$. There exists an L-shaped distribution $\mathcal F_k$ such that every Bayesian incentive compatible mechanism for $\mathcal{F}_k$ provides an approximation ratio no better than $1+ \frac{\frac{k}{2}(H_{2k-2}-H_{k-1})}{2(1-\frac{1}{2^k}+\frac{k}{2})}$ to the optimal welfare, where $H_{n}=\sum_{k=1}^n \frac{1}{k}$ is the Harmonic number.
    This ratio approaches $1+ \frac{\ln{2}}{2}$ as $k$ approaches $\infty$.
\end{theorem}


We choose values for the parameters of the distribution $\mathcal{F}_k$ such that $\mathcal{F}_k$ is in standard $L$-shaped form (Definition~\ref{standard-L-distribution}). Let $\varepsilon > 0$ be a small enough value, we set the values:
\begin{equation}\label{k-k-welfare-prob-values}
    \prbuyervalue{i} = \begin{cases}
                 \frac{1}{2^i} &  1 \leq i \leq k-1 ;\\
                 \frac{1}{2^{k-1}} & i = k.
           \end{cases} \quad
    q_i = \begin{cases}
                 \frac{1}{2}  &   i = 1 ;\\
                 \frac{1}{2(k-1)}  & 1 < i \leq k.
           \end{cases} \quad 
            b_i = \begin{cases}
                 2^{i-1}\frac{k}{k-1+i} &  1 \leq i \leq k-1 ;\\
                 \frac{1-\prbuyervalue{k}+q_1\prbuyervalue{k} + \prbuyervalue{1}(\frac{q_1+q_2}{q_2})}{q_1\prbuyervalue{k}} + \varepsilon & i = k.
           \end{cases}           
\end{equation}
Next, we select the values for $s_2, \dots s_k$ such that Condition~\ref{same-welfare-cond} in Definition~\ref{standard-L-distribution} is met, for every $1 \leq i \leq k-1$:
\begin{equation}\label{k-k-welfare-s-values}
\begin{split}
    & \Pr(0, b_i)\cdot b_i = \Pr(s_{i+1}, b_k) \cdot (b_k -s_{i+1}) \iff \prbuyervalue{i}b_i = \prbuyervalue{k}q_{i+1}\cdot(b_k -s_{i+1}) \\
    & \iff s_j = b_k -\frac{\prbuyervalue{j-1}b_{j-1}}{\prbuyervalue{k}\cdot q_j} \quad  \forall 2 \leq j \leq k.
    \end{split}
\end{equation}

\begin{lemma}\label{cor-choice-of-pars-lemma}
    For the choice of parameters in Equation~\ref{k-k-welfare-prob-values} and Equation~\ref{k-k-welfare-s-values}, the first condition of Definition~\ref{standard-L-distribution} holds.
\end{lemma}

By Lemma~\ref{k-k-paramerts-fit-dist}, our choice of values indeed  yields a legal distribution, and the order over the values is as assumed throughout the proof, i.e., $b_1 < b_2 < \dots < b_k$, $s_1 < s_2 < \dots < s_k$, and $s_k < b_k$, 

\begin{lemma}\label{k-k-paramerts-fit-dist}
     For the choice of parameters in Equation~\ref{k-k-welfare-prob-values} and Equation~\ref{k-k-welfare-s-values}, we get a joint distribution over the values of the buyer and seller, and $b_1 < \dots < b_k$, $s_1 < \dots < s_k$, and $s_k < b_k$ and $k\geq 2$. 
\end{lemma}
We will prove these lemmas after using them to bound the approximation ratio $\frac{\welfare{OPT}{\mathcal{F}_k}}{W_{k}}$. 
\begin{align*}
    \frac{\welfare{OPT}{\mathcal{F}_k}}{W_{k}}  = \frac{\prbuyervalue{1} + \prbuyervalue{2}b_2 +\dots \prbuyervalue{k}b_{k}}{\prbuyervalue{k}b_k} =& 1+ \frac{\prbuyervalue{1} + \prbuyervalue{2}b_2 +\dots \prbuyervalue{k-1}b_{k-1}}{\prbuyervalue{k}b_k} = 1+ \frac{\sum_{i=0}^{k-2}\frac{k}{2(k+i)}}{2(1-\frac{1}{2^k}+\frac{k}{2}+ \varepsilon)}\\
     =& 1+ \frac{\frac{k}{2}(H_{2k-2}-H_{k-1})}{2(1-\frac{1}{2^k}+\frac{k}{2}+ \varepsilon)}.
\end{align*}
As $\varepsilon$ approaches $0$, the bound $\frac{\welfare{OPT}{\mathcal{F}_k}}{W_{k}}$ approaches $1+ \frac{\frac{k}{2}(H_{2k-2}-H_{k-1})}{2(1-\frac{1}{2^k}+\frac{k}{2})}$. As $k$ approaches $\infty$, $\frac{\welfare{OPT}{\mathcal{F}_k}}{W_{k}}$ approaches $1+\frac{\ln{2}}{2}$, because $\lim_{n \to \infty}(H_n -\ln{n})=\gamma$, where $\gamma$ is Euler's constant. This proves the theorem assuming Lemmas \ref{cor-choice-of-pars-lemma} and \ref{k-k-paramerts-fit-dist}. We now provide the proofs of these lemmas.


\begin{proof}[Proof of Lemma~\ref{cor-choice-of-pars-lemma}]
    The first condition of Definition~\ref{L-shaped-k-k} holds for every $1 \leq i\leq k-1$ and $i+1 \leq j < k $ if:
     \begin{align*}
     &  \quad \frac{\sum_{r=2}^j\Pr\left(s=s_r| b=b_k\right)\left(b_k-s_j\right)}{\Pr\left(s=0|b=b_k\right)} +b_i  < s_j -\frac{b_i\cdot \sum_{r=i}^{k-1}\Pr\left(b=b_r| s=0\right)}{\Pr\left(b=b_k|s=0\right)}\\
      \iff &\quad  \left(b_k-s_j\right)\frac{\left(q_2 + \dots q_j\right)}{q_1}+ b_i < s_j - \frac{b_i\left(\prbuyervalue{i} +\dots \prbuyervalue{k-1}\right)}{\prbuyervalue{k}q_1} \\
     \underbrace{\iff}_{~\eqref{k-k-welfare-s-values}} & \quad b_i\left(\frac{q_1 \prbuyervalue{k}+ \prbuyervalue{i} + \dots + \prbuyervalue{k-1}}{q_1\prbuyervalue{k}}\right) + \frac{\prbuyervalue{j-1}b_{j-1}}{\prbuyervalue{k}q_j}\left(\frac{q_1+\dots +q_j}{q_1}\right) < b_k \\
    \iff &\quad  b_i\left(q_1\prbuyervalue{k} + \prbuyervalue{i} + \dots + \prbuyervalue{k-1}\right) + \frac{\prbuyervalue{j-1}b_{j-1}}{q_j}\left(q_1+\dots +q_j\right)   <  b_k\cdot q_1\cdot \prbuyervalue{k}\\
    \underbrace{\iff}_{~\eqref{k-k-welfare-prob-values}} &\quad b_i\left(q_1\prbuyervalue{k} + \prbuyervalue{i} + \dots + \prbuyervalue{k-1}\right) + \frac{\prbuyervalue{j-1}b_{j-1}}{q_j}\left(q_1+\dots +q_j\right) <  1-\prbuyervalue{k}+q_1\prbuyervalue{k} + \prbuyervalue{1}\left(\frac{q_1+q_2}{q_2}\right) + \varepsilon \cdot q_1\prbuyervalue{k}.      
    \end{align*}
    Then, it is enough to show that for every $1 \leq i\leq k-1$ and $i+1 \leq j < k $, the two inequalities hold:
    \begin{subequations}
    \begin{equation}\label{cor-lemma-choice-of-pars-first}
        b_i\left(q_1\prbuyervalue{k} + \prbuyervalue{i} + \dots + \prbuyervalue{k-1}\right) \leq 1-\prbuyervalue{k}+q_1\prbuyervalue{k}
    \end{equation}
    \begin{equation}\label{cor-lemma-choice-of-pars-second}
        \frac{\prbuyervalue{j-1}b_{j-1}}{q_j}\left(q_1+\dots +q_j\right) \leq   \prbuyervalue{1}\left(\frac{q_1+q_2}{q_2}\right)
    \end{equation}
    \end{subequations}
    We start with the first Inequality~\ref{cor-lemma-choice-of-pars-first} using our values for the parameters (Equation~\ref{k-k-welfare-prob-values} and Equation~\ref{k-k-welfare-s-values}:
    \begin{align*}
         & \quad b_i\left(q_1\prbuyervalue{k} + \prbuyervalue{i} + \dots + \prbuyervalue{k-1}\right) \leq 1-\prbuyervalue{k}+q_1\prbuyervalue{k}\\
         \iff &\quad b_i\left(\frac{1}{2^k}+ \frac{1}{2^i}+ \dots + \frac{1}{2^{k-1}}\right)  \leq 1-\frac{1}{2^k} \\ 
         \iff &\quad  b_i \left(\frac{1}{2^{i-1}}-\frac{1}{2^k}\right)   \leq 1-\frac{1}{2^k} \\ 
         \underbrace{\iff}_{~\eqref{k-k-welfare-prob-values}} & \quad \left(2^{i-1}\frac{k}{k-1+i}\right) \left(\frac{1}{2^{i-1}}-\frac{1}{2^k}\right) \leq 1-\frac{1}{2^k}\\
         \iff &\quad \frac{k}{k-1+i} \left(1- \frac{1}{2^{k-i+1}}\right)  \leq 1-\frac{1}{2^k}.
    \end{align*}
    Since $i\geq 1$ we have that $\frac{k}{k-1+i} \leq 1, -\frac{1}{2^{k-i+1}} \leq -\frac{1}{2^k}$. Thus, $\frac{k}{k-1+i} \left(1- \frac{1}{2^{k-i+1}}\right) \leq 1- \frac{1}{2^k}$, as needed.
    We are left with proving Inequality~(\ref{cor-lemma-choice-of-pars-second}). By recalling that $j\geq i+1 \geq 2$ and our choice of values (Equations (\ref{k-k-welfare-prob-values}) and (\ref{k-k-welfare-s-values})):
        \begin{align*}
        \frac{\prbuyervalue{j-1}b_{j-1}}{q_j}\left(q_1+\dots +q_j\right) &\leq   \prbuyervalue{1}\left(\frac{q_1+q_2}{q_2}\right) \iff b_{j-1} \leq \prbuyervalue{1}\left(\frac{q_1+q_2}{q_2}\right)\cdot \frac{q_j}{\prbuyervalue{j-1}\left(q_1+\dots +q_j\right)}\\
        \underbrace{\iff}_{j\geq 2} \,  b_{j-1} & \leq \frac{\frac{1}{2}\left(\frac{1}{2}+ \frac{1}{2(k-1)}\right)}{\frac{1}{2^{j-1}}\left(\frac{1}{2}+\frac{j-1}{2\left(k-1\right)}\right)} = \frac{2^{j-2}\left(\frac{k}{k-1}\right)}{\frac{k-1+j-1}{k-1}}=2^{j-2}\left(\frac{k}{k-1+j-1}\right) = b_{j-1}.
    \end{align*}

\end{proof}

\begin{proof}[Proof of Lemma~\ref{k-k-paramerts-fit-dist}]

We first verify that the buyer's marginal probability distribution sums up to 1: $\sum_{i=1}^k \prbuyervalue{i} = \sum_{i=1}^{k-1}\frac{1}{2^i} + \frac{1}{2^{k-1}} = 1$. Moreover, for every value in the buyer's support, the conditional distribution of the seller also sums up to $1$. For $b \in \{ b_1, \dots, b_{k-1} \}$, the seller's value is $0$ with probability $1$ and for $b=b_k$, the sum of the seller's conditional probabilities is $\sum_{i=1}^k q_{i} = \frac{1}{2} + (k-1)\frac{1}{2(k-1)} = 1$. 

Next, we verify that indeed $b_i<b_{i+1}$, for every $1 \leq i \leq k-1$. 
For $1 \leq i \leq k-2$ we have:
\begin{align*}
    b_i < b_{i+1} \iff 2^{i-1}\frac{k}{k-1+i} < 2^{i} \cdot \frac{k}{k-1+i+1} \iff k+i < 2(k+i -1) \iff 2<k+i,
\end{align*}
and for $k\geq 2$, the last inequality holds. For $i=k-1$, we get:
\begin{align*}
    b_k > b_{k-1} & \iff  \frac{1-\prbuyervalue{k}+q_1\prbuyervalue{k} + \prbuyervalue{1}\left(\frac{q_1+q_2}{q_2}\right)}{q_1\prbuyervalue{k}} > 2^{k-2} \cdot \frac{k}{k-1+k-1}\\ 
    & \iff \frac{1-\frac{1}{2^{k-1}}+\frac{1}{2^k} +\frac{1}{2}\left(\frac{\frac{1}{2}+\frac{1}{2(k-1)}}{\frac{1}{2(k-1)}}\right)}{\frac{1}{2^k}}> 2^{k-2} \cdot \frac{k}{k-1+k-1} \iff 4\left(1-\frac{1}{2^k}+\frac{k}{2}\right) > \frac{k}{2(k-1)}.
\end{align*}
where the last inequality holds for $k\geq 2$. Next, we show that $s_i < s_{i+1}$, for every $1\leq i \leq k-1$. 
For $i =1$, we have:
\begin{align*}
    s_1 < s_2 & \iff 0 < b_k- \frac{\prbuyervalue{1}b_1}{\prbuyervalue{k}q_2} +\varepsilon \iff 0< \frac{2\left(1-\frac{1}{2^{k-1}}+\frac{1}{2^k} +\frac{1}{2}\left(\frac{\frac{1}{2}+\frac{1}{2(k-1)}}{\frac{1}{2(k-1)}}\right)\right)-\frac{1}{2}2(k-1)}{\prbuyervalue{k}} +\varepsilon\\
    & \iff 0< \frac{2\left(1-\frac{1}{2^k}+\frac{k}{2}\right)-k+1}{\prbuyervalue{k}} +\varepsilon.
\end{align*}
and the inequality clearly holds. For $2 \leq i \leq k-1$, we get:
\begin{align*}
    s_i < s_{i+1} & \iff  b_k- \frac{\prbuyervalue{i-1}b_{i-1}}{\prbuyervalue{k}q_i} < b_k- \frac{\prbuyervalue{i}b_i}{\prbuyervalue{k}q_{i+1}} \iff \frac{\prbuyervalue{i}b_i}{q_{i+1}} < \frac{\prbuyervalue{i-1}b_{i-1}}{q_{i}} \underbrace{\iff}_{q_i=q_{i+1} \text{ for } i>1} \prbuyervalue{i}b_i < \prbuyervalue{i-1}b_{i-1}\\
    & \iff \frac{k}{2(k-1+i)} < \frac{k}{2(k-1+i-1)} \iff k+i -2 < k+i-1 \iff -2< -1.
\end{align*}

Finally, we are left with proving $s_k < b_k$:
\begin{align*}
    s_k < b_k \iff b_k -\frac{\prbuyervalue{k-1}b_{k-1}}{\prbuyervalue{k}q_k} < b_k \iff 0 < \frac{\prbuyervalue{k-1}b_{k-1}}{\prbuyervalue{k}q_k},
\end{align*}
and the last inequality clearly holds, as all values are strictly positive.
\end{proof}

 \subsection{ Approximating the Gains from Trade with Independent Values}\label{gft-section}

This section proves a lower bound for approximating the gains from trade.
We consider a family of independent distributions that are $L$-shaped (Section~\ref{L-shape-sec}). 
In every distribution $\mathcal{F}_k$ (depicted in Figure~\ref{ind:k-k-gft}) in the family, the buyer's support is $b_1=1< b_2< \dots< b_k$, and the seller's support is $S_1=0 < s_2< \dots < s_k$.
The probability of the buyer's value being $b_i$ is $\prbuyervalue{i}$ for $i \in [k]$, and the probability that the seller's value is $s_i$ is $q_i$, for $i \in [k]$. In every instance that is not in the ``$L$'', i.e., in every instance in the set
${\overline{L}} = \{ (s_j, b_i) | s_j \neq 0 \text{ and } b_i \neq b_k\}$, the seller's value is is larger than the buyer's value ( $s_2 > b_{k-1}$) and trade cannot occur.

\begin{theorem} \label{k-k-cor-gft-thm}
    Let $k\geq 2$. There exists an L-shaped distribution $\mathcal F_k$ such that every Bayesian incentive compatible mechanism for $\mathcal{F}_k$ provides an approximation ratio no better than $1+ \frac{H_{k}-1}{H_{k} +1}$ to the optimal gains from trade, where $H_{n}=\sum_{k=1}^n \frac{1}{k}$ is the Harmonic number.
    This ratio approaches $2$ as $k$ approaches $\infty$.
\end{theorem}
This ratio is tight for $L$-shaped distributions. To see this, consider two allocation rules: the first trades the item when $s=0$, and the second trades the item when $b=b_k$. Both allocation rules can be implemented by fixed price mechanisms: the first sets a price of $0$ and the second a price of $s_k$. The sum of the gains from trade of the two mechanisms is the optimal gains from trade, so at least one of them guarantees a ratio of $2$.    

\begin{figure}[H]
    \centering
    \renewcommand{\arraystretch}{1.4}
\begin{NiceTabularX}{12cm}{YYYYYY}[hvlines]
\CodeBefore
  \begin{tikzpicture}
    \cellcolor{tradepossiblecolor}{2-2,3-2, 4-2, 5-2, 6-2, 6-3, 6-4, 6-5, 6-6}
  \end{tikzpicture}
\Body
    \diagbox{ buyer}{seller } &$s_1 = 0$  & $s_2$  & $\dots$ & $s_{k-1}$ & $s_k$  \\
    $b_1 = 1$ &  ${\prbuyervalue{1}} \cdot q_1^*$ & ${\prbuyervalue{1}}\cdot q_2$ & $\dots$ & ${\prbuyervalue{1}} \cdot q_{k-1}$ & ${\prbuyervalue{1}} \cdot q_k$\\
    $b_2$ & ${\prbuyervalue{2}}\cdot q_1^*$ & ${\prbuyervalue{2}}\cdot q_2$ & $\dots$ & ${\prbuyervalue{2}}\cdot q_{k-1}$ & ${\prbuyervalue{2}}\cdot q_k$\\
    $\vdots$ &  ${\vdots}$ & $\vdots$ & $\dots$ & $\vdots$ & $\vdots$\\
    $b_{k-1}$ &   ${\prbuyervalue{k-1}}\cdot q_1 ^*$ & ${\prbuyervalue{k-1}}\cdot q_2$ & $\dots$ & ${\prbuyervalue{k-1}}\cdot q_{k-1} $ & ${\prbuyervalue{k-1}}\cdot q_k$\\
    $b_k$ &   ${\prbuyervalue{k}\cdot q_1}^*$ & ${\prbuyervalue{k}\cdot q_2}^*$ & $\dots$ & ${\prbuyervalue{k}\cdot q_{k-1}}^*$ & ${\prbuyervalue{k}\cdot q_k}^*$\\
    \end{NiceTabularX}
    \caption{An independent L-shaped distribution and the allocation function that maximizes its gains from trade. Thecolored cellsd in \tradepossiblecolorname correspond to instances in which a trade can occur.}
    \label{ind:k-k-gft}
\end{figure}

We choose values for the parameters of the distribution $\mathcal{F}_k$ such that $\mathcal{F}_k$ is in standard $L$-shaped form (Definition~\ref{standard-L-distribution}). Let $\varepsilon > 0$ be small enough. We set the values:
\begin{equation}\label{k-k-gft-prob-values}
    \prbuyervalue{i} = \begin{cases}
                 \frac{1}{2^i} &  1 \leq i \leq k-1 ;\\
                 \frac{1}{2^{k-1}} & i = k.
           \end{cases} \quad
    q_i = \frac{1}{k} \quad 
            b_i = \begin{cases}
                 \frac{2^{i}}{1+i} &  1 \leq i \leq k-1 ;\\
                 \frac{2}{\prbuyervalue{k}} + \varepsilon & i = k.
           \end{cases}           
\end{equation}

Next, we select the values for $s_2, \dots s_k$ such that Condition~\ref{same-welfare-cond} in Definition~\ref{standard-L-distribution} is met, for every $1\leq i\leq k-1$:
\begin{equation}\label{k-k-gft-s-values}
\begin{aligned}
    & \quad \Pr(0, b_i)\cdot b_i = \Pr(s_{i+1}, b_k) \cdot (b_k -s_{i+1})\\
    \iff & \quad \prbuyervalue{i}\cdot q_1 \cdot b_i = \prbuyervalue{k}\cdot q_{i+1} \cdot (b_k -s_{i+1})\\
    \iff & \quad s_{i+1} = b_k -\frac{b_i\cdot\prbuyervalue{i}\cdot q_1}{\prbuyervalue{x}\cdot q_{i+1}}\\
    \underbrace{\iff}_{~\eqref{k-k-gft-prob-values}} & \quad s_{i+1} = b_k -\frac{b_i\cdot\prbuyervalue{i}}{\prbuyervalue{x}} \\
    \iff & \quad  s_j = b_k -\frac{\prbuyervalue{j-1}\cdot b_{j-1}}{\prbuyervalue{k}} \quad  \forall 2 \leq j \leq k.
\end{aligned}
\end{equation}

\begin{lemma}\label{cor-choice-of-pars-lemma-gft}
    For the choice of parameters in Equations~\ref{k-k-gft-prob-values} and \ref{k-k-gft-s-values}, the first condition of Definition~\ref{standard-L-distribution} holds.
\end{lemma}


\begin{lemma}\label{k-k-paramerts-fit-dist-gft}
     For the choice of parameters in Equation~\ref{k-k-gft-prob-values} and Equation~\ref{k-k-gft-s-values}, we get a joint distribution over the values of the buyer and seller, and $b_1 < \dots < b_k$, $s_1 < \dots < s_k$, $s_2 > b_{k-1}$, $s_k < b_k$, for $k\geq 2$. 
\end{lemma}

Observe that these two lemmas and our choice of parameters in Equation~\ref{k-k-gft-s-values} imply a family of $L$-shaped distributions in standard form (Definition~\ref{L-shaped-k-k}). Thus, we can apply Claim~\ref{standard-form-claim} and analyze the  
approximation ratio $\frac{\welfare{OPT}{\mathcal{F}_k}}{GFT_{k}}$. Recall that in an $L$-shaped distribution, the only instances in which trade can occur are in the ``$L$''. Thus, instances in the ``$L$'' are the only instances that affect the gains from trade.
We now analyze the approximation ratio and then prove the two lemmas.

By Claim~\ref{standard-form-claim}, the gains from trade of every mechanism $(b_i, s_j, 1)$ that trades the item in $k$ instances is the same, and we denoted it by $GFT_k$. The allocation rule $(b_1, 0, 1)$ trades the item only in the instances where the seller's value is $0$, and thus its gains from trade is $q_1\sum_{i=1}^k \prbuyervalue{i}b_i$, and $GFT_k = q_1\sum_{i=1}^k \prbuyervalue{i}b_i$.

\begin{align*}
    \frac{\welfare{OPT}{\mathcal{F}_k}}{GFT_{k}}  & = \frac{q_1\cdot \sum_{i=1}^k \prbuyervalue{i}\cdot b_i + \prbuyervalue{k}\cdot \sum_{j=2}^kq_j\cdot (b_k - s_j)}{q_1\cdot \sum_{i=1}^k \prbuyervalue{i}b_i} = 1+ \frac{\prbuyervalue{k}\sum_{j=2}^kq_j(b_k-b_k+\frac{\prbuyervalue{j-1}b_{j-1}}{\prbuyervalue{k}})}{q_1\cdot \sum_{i=1}^k \prbuyervalue{i}b_i}\\
    & = 1+ \frac{q_1\cdot \sum_{i=1}^{k-1} \prbuyervalue{i}\cdot b_i}{q_1\cdot \sum_{i=1}^k \prbuyervalue{i}\cdot b_i} = 1+ \frac{\sum_{i=1}^k\frac{1}{1+i}}{\sum_{i=1}^{k-1}\frac{1}{1+i} + 2+ \varepsilon\cdot \prbuyervalue{k}}
     = 1+ \frac{H_{k}-1}{H_k -1 +2 +\varepsilon \cdot \prbuyervalue{k}}.  
\end{align*}
As $\varepsilon$ approaches $0$, the bound $\frac{\welfare{OPT}{\mathcal{F}_k}}{GFT_{k}}$ approaches $ 1+ \frac{H_{k}-1}{H_k -1 +2} $, and as $k$ approaches $\infty$, $\frac{\welfare{OPT}{\mathcal{F}_k}}{GFT_{k}}$ approaches $2$, due to the fact that $\lim_{n \to \infty}(H_n -\ln{n})=\gamma$, where $\gamma$ is Euler's constant. This proves the theorem assuming Lemmas \ref{cor-choice-of-pars-lemma-gft} and \ref{k-k-paramerts-fit-dist-gft}. We now provide the proofs of these lemmas.


\begin{proof}[Proof of Lemma~\ref{cor-choice-of-pars-lemma-gft}]
The first condition of Definition~\ref{L-shaped-k-k} is that for every $1\leq i\leq k-1$ and every $i+1 \leq j<k$: 
    $$
    \frac{\sum_{r=2}^j\Pr(s=s_r| b=b_k)(b_k-s_j)}{\Pr(s=0|b=b_k)} +b_i < s_j -\frac{b_i\cdot \sum_{r=i}^{k-1}\Pr(b=b_r| s=0)}{\Pr(b=b_k|s=0)}.
    $$ 
    For our independent distributions, it is equal to:
    \begin{align*}
         b_i + (b_k-s_j)\cdot \frac{(q_2+\dots +q_j)}{q_1} < & s_j - b_i\cdot (\frac{\prbuyervalue{i}+ \dots + \prbuyervalue{k-1}}{\prbuyervalue{x}})\\ & \underbrace{\iff}_{~\eqref{k-k-gft-s-values}} \\
        b_i(\prbuyervalue{k} + \prbuyervalue{i} + \dots + \prbuyervalue{k-1}) + \frac{\prbuyervalue{j-1}b_{j-1}}{q_j}(q_1+\dots +q_j)  <& b_k\cdot \prbuyervalue{k}  \underbrace{=}_{~\eqref{k-k-gft-prob-values}} 2 +\prbuyervalue{k}\varepsilon.        
    \end{align*}
    Then, it is enough to show that for every $1 \leq i\leq k-1$ and $i+1 \leq j < k $, the two inequalities hold:
    \begin{subequations}
    \begin{equation}\label{cor-lemma-choice-of-pars-first-gft}
        b_i\cdot (\prbuyervalue{k} + \prbuyervalue{i} + \dots + \prbuyervalue{k-1}) \leq 1
    \end{equation}
    \begin{equation}\label{cor-lemma-choice-of-pars-second-gft}
        \frac{\prbuyervalue{j-1}\cdot b_{j-1}} {q_j}\cdot (q_1+\dots +q_j) \leq   1
    \end{equation}
    \end{subequations}
    We start with the first Inequality~\ref{cor-lemma-choice-of-pars-first-gft} using our values for the parameters (Equation~\ref{k-k-gft-prob-values} and Equation~\ref{k-k-gft-s-values}):
    \begin{align*}
          b_i\cdot (\prbuyervalue{k} + \prbuyervalue{i} + \dots + \prbuyervalue{k-1}) = \frac{2^i}{1+i}\cdot \frac{1}{2^{i-1}} = \frac{2}{i+1} \leq 1\iff& 2\leq i+1 \iff 1 \leq i.
    \end{align*}
    We now prove the second Inequality~\ref{cor-lemma-choice-of-pars-second-gft}, for every $2\leq i+1\leq j \leq k-1$:
    \begin{align*}
         \frac{\prbuyervalue{j-1}\cdot b_{j-1}}{q_j}\cdot (q_1+\dots +q_j) = \frac{1}{j-1+1}\cdot \frac{\frac{j}{k}}{\frac{1}{k}} = 1 \leq   1.
    \end{align*}
\end{proof}

\begin{proof}[Proof of Lemma~\ref{k-k-paramerts-fit-dist-gft}]
We first verify that the buyer's marginal probability distribution sums up to 1: $\sum_{i=1}^k \prbuyervalue{i} = \sum_{i=1}^{k-1}\frac{1}{2^i} + \frac{1}{2^{k-1}} = 1$, and the seller's marginal distribution sums up to 1: $\sum_{i=1}^k q_i = \frac{1}{k}\cdot k =1$ for the values chosen in Equation~\ref{k-k-gft-prob-values}.

Next, we verify that indeed $b_i<b_{i+1}$, for every $1 \leq i \leq k-1$. 
For $1 \leq i \leq k-2$ we have:
\begin{align*}
    b_i < b_{i+1} \iff \frac{2^i}{i+1} < \frac{2^{i+1}}{i+1+1} \iff i+2 < 2(i+1) \iff i< 2i \iff 1<2.
\end{align*}
For $i=k-1$, we get:
\begin{align*}
    b_{k-1}< b_k \iff \frac{2^{k-1}}{k-1+1} < \frac{2}{\prbuyervalue{k}}+\varepsilon \iff \frac{2^{k-1}}{k} < \frac{2}{\frac{1}{2^{k-1}}} +\varepsilon \iff \frac{2^{k-1}}{k} < 2^k +\varepsilon. 
\end{align*}
Next, we show that $s_i < s_{i+1}$, for every $1\leq i \leq k-1$. 
For $i =1$, we have:
\begin{align*}
    s_1 < s_2 & \iff 0 < \frac{2+ \varepsilon\cdot \prbuyervalue{k} - \frac{1}{1+1}}{\prbuyervalue{k}} = \frac{1.5}{\prbuyervalue{k}}+ \varepsilon.
\end{align*}
and the inequality clearly holds. For $2 \leq i \leq k-1$, we get:
\begin{align*}
    s_i < s_{i+1} & \iff  \frac{2+ \varepsilon\cdot \prbuyervalue{k} - \frac{1}{i+1-1}}{\prbuyervalue{k}}< \frac{2+ \varepsilon\cdot \prbuyervalue{k} - \frac{1}{i+1-1+1}}{\prbuyervalue{k}} \iff -\frac{1}{i}< -\frac{1}{i+1} \iff \frac{1}{i} > \frac{1}{i+1} \iff 1> 0.
\end{align*}

Next, we show that $s_2 > b_{k-1}$:
\begin{align*}
     \frac{2+ \varepsilon\cdot \prbuyervalue{k} - \frac{1}{2}}{\prbuyervalue{k}} > \frac{2^{k-1}}{k-1+1} \iff \varepsilon + (\frac{3}{2})\cdot 2^{k-1} > \frac{2^{k-1}}{k}.
\end{align*}
Then it is enough to show that $\frac{3}{2} > \frac{1}{k}$, which holds for $k > \frac{2}{3}$. Finally, we are left with proving $s_k < b_k$:
\begin{align*}
    s_k < b_k \iff \frac{2+\varepsilon\cdot\prbuyervalue{x} -\frac{1}{k-1+1}}{\prbuyervalue{k}} < b_k \iff 2+\varepsilon \cdot \prbuyervalue{k}-\frac{1}{k} < 2+\varepsilon \cdot \prbuyervalue{k} \iff -\frac{1}{k} < 0.
\end{align*}
and the last inequality clearly holds as $k \geq 1$.
\end{proof}

\bibliography{bilateral-trade}

\begin{thebibliography}{CBdKLT16}

\bibitem[BCGZ18]{BothWorlds}
Moshe Babaioff, Yang Cai, Yannai~A. Gonczarowski, and Mingfei Zhao.
\newblock The best of both worlds: Asymptotically efficient mechanisms with a
  guarantee on the expected gains-from-trade.
\newblock In {\em Proceedings of the 2018 ACM Conference on Economics and
  Computation}, EC '18, page 373, New York, NY, USA, 2018. Association for
  Computing Machinery.

\bibitem[BCWZ17]{Bru}
Johannes Brustle, Yang Cai, Fa~Wu, and Mingfei Zhao.
\newblock Approximating gains from trade in two-sided markets via simple
  mechanisms.
\newblock In {\em Proceedings of the 2017 ACM Conference on Economics and
  Computation}, EC '17, page 589–590, New York, NY, USA, 2017. Association
  for Computing Machinery.

\bibitem[BD14]{BD14}
Liad Blumrosen and Shahar Dobzinski.
\newblock Reallocation mechanisms.
\newblock In {\em Proceedings of the Fifteenth ACM Conference on Economics and
  Computation}, EC '14, page 617, New York, NY, USA, 2014. Association for
  Computing Machinery.

\bibitem[BD21]{BD21}
Liad Blumrosen and Shahar Dobzinski.
\newblock (almost) efficient mechanisms for bilateral trading.
\newblock {\em Games and Economic Behavior}, 130:369--383, 2021.

\bibitem[BDK21]{DKB}
Moshe Babaioff, Shahar Dobzinski, and Ron Kupfer.
\newblock A note on the gains from trade of the random-offerer mechanism, 2021.

\bibitem[BM16]{BM}
Liad Blumrosen and Yehonatan Mizrahi.
\newblock Approximating gains-from-trade in bilateral trading.
\newblock In Yang Cai and Adrian Vetta, editors, {\em Web and Internet
  Economics}, pages 400--413, Berlin, Heidelberg, 2016. Springer Berlin
  Heidelberg.

\bibitem[CBdKLT16]{italians}
Riccardo Colini-Baldeschi, Bart de~Keijzer, Stefano Leonardi, and Stefano
  Turchetta.
\newblock Approximately efficient double auctions with strong budget balance.
\newblock In {\em Proceedings of the Twenty-Seventh Annual ACM-SIAM Symposium
  on Discrete Algorithms}, SODA '16, page 1424–1443, USA, 2016. Society for
  Industrial and Applied Mathematics.

\bibitem[CW23]{CW-STOC}
Yang Cai and Jinzhao Wu.
\newblock On the optimal fixed-price mechanism in bilateral trade.
\newblock In {\em Proceedings of the 55th Annual ACM Symposium on Theory of
  Computing}, STOC 2023, page 737–750, New York, NY, USA, 2023. Association
  for Computing Machinery.

\bibitem[DMSW22]{GDTCA}
Yuan Deng, Jieming Mao, Balasubramanian Sivan, and Kangning Wang.
\newblock Approximately efficient bilateral trade.
\newblock In {\em Proceedings of the 54th Annual ACM SIGACT Symposium on Theory
  of Computing}, STOC 2022, page 718–721, New York, NY, USA, 2022.
  Association for Computing Machinery.

\bibitem[Fei22]{IGFT}
Yumou Fei.
\newblock Improved approximation to first-best gains-from-trade.
\newblock In Kristoffer~Arnsfelt Hansen, Tracy~Xiao Liu, and Azarakhsh
  Malekian, editors, {\em Web and Internet Economics}, pages 204--218, Cham,
  2022. Springer International Publishing.

\bibitem[KPV22]{KPV}
Zi~Yang Kang, Francisco Pernice, and Jan Vondrák.
\newblock Fixed-price approximations in bilateral trade.
\newblock In {\em Proceedings of the 2022 Annual ACM-SIAM Symposium on Discrete
  Algorithms (SODA)}, pages 2964--2985, Philadelphia, PA, 2022. Society for
  Industrial and Applied Mathematics.

\bibitem[LRW23]{CSTOC}
Zhengyang Liu, Zeyu Ren, and Zihe Wang.
\newblock Improved approximation ratios of fixed-price mechanisms in bilateral
  trades.
\newblock In {\em Proceedings of the 55th Annual ACM Symposium on Theory of
  Computing}, STOC 2023, page 751–760, New York, NY, USA, 2023. Association
  for Computing Machinery.

\bibitem[Mal22]{malik22}
Komal Malik.
\newblock Optimal robust mechanism in bilateral trading, 2022.

\bibitem[McA92]{TradeReduction}
R.Preston McAfee.
\newblock A dominant strategy double auction.
\newblock {\em Journal of Economic Theory}, 56(2):434--450, 1992.

\bibitem[McA08]{mcaffee}
R~Preston McAfee.
\newblock The gains from trade under fixed price mechanisms.
\newblock {\em Applied Economics Research Bulletin}, 1(1):1--10, 2008.

\bibitem[Mot21]{corThesis}
Daniel Motoc.
\newblock The bilateral trade problem: The case of correlated values.
\newblock Undergraduate thesis, University of Michigan, 2021.

\bibitem[MR92]{MR}
R.~Preston McAfee and Philip~J. Reny.
\newblock Correlated information and mecanism design.
\newblock {\em Econometrica}, 60(2):395--421, 1992.

\bibitem[MS83]{Myerson-Satterthwaite}
Roger~B Myerson and Mark~A Satterthwaite.
\newblock Efficient mechanisms for bilateral trading.
\newblock {\em Journal of Economic Theory}, 29(2):265--281, 1983.

\end{thebibliography}
\bibliographystyle{alpha}

\appendix

\section{The Lack of Power of Dominant-Strategy Mechanisms} \label{sec-dominant-strategy-lb}

Consider the discrete joint distribution $\distributionlbdsic{k}$.
In $\distributionlbdsic{k}$, the set of possible values of the buyer is $\{0, k, k^2, \dots ,k^{k}\}$ and the set of possible values of the seller is $\{0,1,\ldots, k^{k-1}\}$. The probability that an instance $(s,b)$ occurs is:

$$\distributionlbdsicpdf{k}(s,b) = \begin{cases}
        \frac{1}{b} & s=\frac b k, b \in \{k, k^2, \dots ,k^{k}\};\\
        1- \sum_{i=1}^{k} \frac{1}{k^i}-\varepsilon &  b=0, s=0;\\
        \frac{\varepsilon}{k(k+1)} & b \in \{0, k, k^2, \dots ,k^{k}\}, s \in \{0, 1, k, k^2, \dots ,k^{k-1}\}, s\neq \frac{b}{k}, s\neq 0 \text{ or } b \neq 0 ;\\
        0 & \text{otherwise.}
        \end{cases}
$$

Where $0<\varepsilon<\frac{1}{k^k}$.

\begin{theorem} \label{unbounded-dsic}
Every dominant strategy incentive compatible mechanism for $\mathcal{F}_{k}$ provides an approximation ratio of $\Omega(k)$ in $\distributionlbdsic{k}$.
\end{theorem}

As discussed in the introduction, the impossibility immediately implies to ``universally truthful'' mechanisms, i.e., probability distributions over dominant-strategy incentive compatible mechanisms (that is, a mechanism that chooses a fixed price at random, which is that way that the state-of-the-art mechanisms for \emph{independent} values \cite{CSTOC,CW-STOC} are stated).

\begin{proof}[Proof of Theorem~\ref{unbounded-dsic}]
    We prove that for every large enough $k \in \mathbb N$, every dominant strategy incentive compatible mechanism for $\mathcal{F}_{k}$ provides an approximation ratio that is no better than $\frac{k}{4}$.
    Recall that every dominant strategy incentive compatible mechanism is a fixed price mechanism \cite{BD14}. Thus, we fix a mechanism $M$ for the distribution $\distributionlbdsic{k}$ and denote its fixed price by $p$.
    
    \begin{lemma}\label{dsic-unbounded-trade-at-one-point}
        There are at most two values of $b \in \{k, k^2, \dots ,{k}^{k}\}$ such that $b \geq p $ and $\frac{b}{k} \leq p$.
    \end{lemma}
    
    \begin{proof}[Proof of Lemma~\ref{dsic-unbounded-trade-at-one-point}]
        Let $b_p\in \{k, k^2, \dots,{k}^{k}\}$ be the smallest value that is at least $p$. If no such value exists, then every value $b\in \{k, k^2, \dots,{k}^{k}\}$ is smaller than $p$, and thus no such value $b$ satisfies the condition in the statement ($b\geq p$), so we are done.
        
        Notice that for any value of $b$ satisfying the condition, it must be that $b\geq b_p$. Therefore, we only consider such values of $b$ from now on. 
        If $\frac{b_p}{k} > p$, then $\frac{b}{k}> p$ for all $b \geq b_p$, and the lemma follows immediately. Hence, assume that $\frac{b_p}{k} \leq p$. For every value of $b \geq b_p \cdot k^2$, we have $\frac{b}{k} \geq b_p \cdot k >p$, so at most two values of $b$ ($b_p$ and $b_p \cdot k$) can satisfy $\frac{b}{k} \leq p$.
    \end{proof}

    Consider an instance $(\frac{b}{k}, b)$ in the support of the distribution $\mathcal{F}_{k}$, with $b \in \{k, k^2, \dots,{k}^{k}\}$, its probability is $\frac{1}{b}$. The contribution of this instance to the expected welfare of a mechanism is $\frac{1}{b}\cdot b = 1$ if the item is traded in this instance, and $\frac{1}{b}\cdot\frac{b}{k} = \frac{1}{k}$ if the item is not traded. The expected welfare of instance where $s\neq \frac{b}{k}$ is at most $\frac{\varepsilon}{k(k+1)}\cdot k^k < \frac{1}{k(k+1)}$, and there are $k(k+1)$  such instances, so it can only increase the expected welfare by $1$.

    Therefore, the optimal welfare of the distribution $\mathcal{F}_{k}$ is at least $k$.
    By Lemma~\ref{dsic-unbounded-trade-at-one-point}, in the mechanism $M$, there are at most two values of $b$ in $ \{k, k^2, \dots ,{k}^{k}\}$ for which the item can be traded. Therefore the expected welfare of $M$ is at most: $2 + \frac{k-2}{k} +1 $ . When $k\geq 3$ the approximation ratio of $M$ is bounded from above by:

    \begin{align*}
        \frac{k}{3+\frac{(k-2)}{k}} = \frac{k}{\frac{4k-2}{k}}= \frac{k^2}{4k-2}>\frac{k}{4}.
    \end{align*}
\end{proof}

\section{A Gap Between Randomized and Deterministic Mechanisms}\label{gap-section}

In this section, we show a gap of $1.113$ between the approximation ratio of a randomized Bayesian incentive compatible mechanism and a deterministic Bayesian incentive compatible mechanism, for approximating the welfare with independent distributions. As far as we know, this is the first proof of a gap between the power of randomized and deterministic mechanisms in the bilateral trade setting.

A randomized mechanism is a mechanism whose allocation function can be randomized, i.e., in every instance instead of simply trading the item or not, it may trade the item with some probability at some possibly random price. 

\begin{claim}
    There is a distribution $\mathcal{F}$ in which the values of the buyer and seller are independent, for which there exists a randomized Bayesian incentive compatible mechanism with approximation ratio of $1.00002$ to the optimal welfare. However, any deterministic Bayesian incentive compatible mechanism has approximation ratio no better than 1.113 to the optimal welfare.  
\end{claim}

\begin{proof}
    Consider the distribution $\mathcal{F}$ of Section~\ref{2-2-ind}. Recall that in this distribution, the buyer has two values, $1$ and $b_2 > 1$. The seller has two values, $0$ and $1 <s_2 < b_2$. With probability $\prbuyervalue{1}$ the buyer's value is $1$ and with probability $q_1$, the seller's value is $0$. 
    The welfare maximizing allocation trades the item in the three instances $(0,1), (0, b_2), (s_2, b_2)$.
    The values of the parameters are: $\prbuyervalue{1} = 0.57, q_1 = 0.716, s_2 = 3.652 > \frac{\prbuyervalue{1}+1}{\prbuyervalue{2}},  b_2 =6.993.$
    The claim follows from the two lemmas below.
    \begin{lemma}\label{randomized-lemma}
        There exists a randomized Bayesian incentive compatible mechanism for $\mathcal{F}$ with an approximation ratio better than $1.00002$.
    \end{lemma}
    \begin{lemma}\label{deterministic-lemma}
        The best deterministic Bayesian incentive compatible mechanism for $\mathcal{F}$ has an approximation ratio no better than $1.113$.
    \end{lemma}

    \begin{proof}[Proof of Lemma~\ref{randomized-lemma}]
        Take some $r$ close to $1$, say $r=0.999$, and consider the following randomized Bayesian incentive compatible mechanism $M=(x,p)$, the fact that it is a randomized mechanism is reflected in the allocation in the instance $(0, 1)$, then the item is traded with probability $r$.
        \begin{equation}\label{rand-alloc-ind-2-2}
        x(s, b) = \begin{cases}
                     r &  s=0, b=1;\\
                     1 & s=0, b=b_2;\\
                     0 & s=s_2, b=1; \\
                     1 & s=s_2, b=b_2.
               \end{cases} \quad
        p(s, b) = \begin{cases}
                     r &  s=0, b=1;\\
                     2.330 & s=0, b=b_2;\\
                     0 & s=s_2, b=1; \\
                     s_2 & s=s_2, b=b_2.
               \end{cases}
        \end{equation}
We prove that this mechanism is Bayesian Incentive compatible by showing that for both the buyer and the seller, for every value in their support, the Bayesian incentive compatibility inequality (\ref{buyer-ic}) holds.
We start with the seller. When his value is $s_2$ and he plays his equilibrium strategy his profit is $\prbuyervalue{1}\cdot s_2 + \prbuyervalue{2}\cdot s_2 =s_2$, which is larger than his profit from playing the equilibrium strategy of $0$, $\prbuyervalue{1}\cdot (r +(1-r)\cdot s_2) +\prbuyervalue{2}(2.330)< s_2$.
When the seller's value is $0$ and he plays his equilibrium strategy, his profit is $\prbuyervalue{1}\cdot r+ \prbuyervalue{2}\cdot 2.33 =1.57133$,  which is larger than his profit when he plays the equilibrium strategy of $s_2$, $\prbuyervalue{2}\cdot s_2 = 1.57036$.

Next, we show that for every value of the buyer in the support, the Bayesian incentive compatibility inequality (\ref{buyer-ic}) holds.
When the buyer's value is $1$ and he plays his equilibrium strategy his profit is $q_1(r-r) = 0$ which is larger than his profit from playing the equilibrium strategy of $b_2$, $q_1\cdot (1-2.33) + q_2\cdot (1-s_2)< 0$.
When the buyer's value is $b_2$ and he plays his equilibrium strategy, his profit is $q_1(b_2-2.33)+ q_2(b_2-s_2) > 4.2875 $, which is larger than his profit when he plays the equilibrium strategy of $1$, $q_1\cdot (b_2\cdot r-r) < 4.2867$.
Hence, $M$ is Bayesian incentive compatible. $M$'s approximation ratio is:
$$ \frac{\prbuyervalue{1} q_1 + \prbuyervalue{1}q_2s_2 + \prbuyervalue{2}b_2 }{r \cdot \prbuyervalue{1} q_1 + \prbuyervalue{1}q_2s_2 + \prbuyervalue{2}b_2} = 1+ \frac{(1-r)\prbuyervalue{1} q_1}{r \cdot \prbuyervalue{1} q_1 + \prbuyervalue{1}q_2s_2 + \prbuyervalue{2}b_2} < 1.00002.$$
    \end{proof}

\begin{proof}[Proof of Lemma~\ref{deterministic-lemma}]
    The proof relies on Section~\ref{2-2-ind}. 
    Recall that for the distribution of Section \ref{2-2-ind}, there are $2^3$ possible allocation functions (as in the instance $(s_2, 1)$ trade cannot occur since $s_2>1$). Observe that any of the four allocation functions that trade the item in at most one instance has smaller welfare compared to at least one of the allocation functions depicted in Figure~\ref{ind:col-allocation} and Figure~\ref{ind:row-allocation}. By Lemma~\ref{ind-no-opt-cond}, if $s_2-\frac{\prbuyervalue{1}}{\prbuyervalue{2}}> \frac{q_2}{q_1}(b_2-s_2) +1$ then the welfare maximizing allocation function (depicted in Figure~\ref{ind:opt-allocation}) is not implementable by a Bayesian incentive compatible mechanism, and by our choice of parameters this is indeed the case: $s_2-\frac{\prbuyervalue{1}}{\prbuyervalue{2}}> 2.3264$, $\frac{q_2}{q_1}(b_2-s_2) +1 < 2.3253$. Similarly, by  Lemma~\ref{ind-no-diag-allocation-cond}, if $s_2-\frac{\prbuyervalue{1}}{\prbuyervalue{2}}> \frac{q_2}{q_1}(b_2-s_2) +1$, then the allocation function depicted in Figure~\ref{ind:diag-allocation} is not implementable by a Bayesian incentive compatible mechanism, and by our choice of parameters this is indeed the case. 
    Finally, the optimal approximation ratio of a deterministic Bayesian incentive compatible for $\mathcal{F}$ is the maximum between the approximation ratio of the allocation rule depicted in Figure~\ref{ind:row-allocation}, and the approximation ratio of the mechanism depicted in Figure~\ref{ind:col-allocation}. Next, we compute the two ratios. 
    The approximation ratio of the allocation rule depicted in Figure~\ref{ind:row-allocation} is:
    $$
    \frac{\prbuyervalue{1}q_1 + x_1\cdot q_2 \cdot s_2 + \prbuyervalue{2}b_2}{ x_1\cdot q_2 \cdot s_2 + \prbuyervalue{2}b_2} = 1+ \frac{\prbuyervalue{1}q_1 }{ x_1\cdot q_2 \cdot s_2 + \prbuyervalue{2}b_2} > 1.1134.
    $$
    The approximation ratio of the allocation rule depicted in Figure~\ref{ind:col-allocation} is:
    $$
    \frac{\prbuyervalue{1}q_1 + x_1\cdot q_2 \cdot s_2 + \prbuyervalue{2}b_2}{\prbuyervalue{1}\cdot q_1+ x_1\cdot q_2 \cdot s_2 + \prbuyervalue{2}\cdot q_1 b_2 + \prbuyervalue{2}q_2\cdot s_2} = 1+ \frac{\prbuyervalue{2}q_2(b_2-s_2)}{\prbuyervalue{1}\cdot q_1+ x_1\cdot q_2 \cdot s_2 + \prbuyervalue{2}\cdot q_1 b_2 + \prbuyervalue{2}q_2\cdot s_2} > 1.1133.
    $$
    Hence, the approximation ratio of any deterministic Bayesian incentive compatible for $\mathcal{F}$ is at least $1.113$.
\end{proof}
\end{proof}

\end{document}